\documentclass[a4paper,11pt]{article}
\newcommand{\ms}{1.8cm}
\usepackage[lmargin=\ms,
                rmargin=\ms,
                tmargin=\ms,
                bmargin=\ms,
                includefoot,
                footskip=2ex]{geometry}

\usepackage[english]{babel}
\usepackage{amsmath,amsthm,amscd,amssymb,amsfonts}
\usepackage[hypertexnames=false]{hyperref}
\usepackage{setspace}
\usepackage{xspace} 

  \vfuzz
  \hfuzz
\ifx\NEWCOMMANDS\undefined
\newcommand{\NEWCOMMANDS}

\newcommand{\mathify}[1]{\ifmmode{#1}\else\mbox{$#1$}\fi}

\newcommand{\ith}{i\textsuperscript{\underline{th}}\xspace}
\newcommand{\jth}{j\textsuperscript{\underline{th}}\xspace}

\newcommand{\fhat} {\mathify{\widehat{f}}\xspace}

\newcommand{\fhatj} {\mathify{\widehat{f^j}}\xspace}
\newcommand{\ghat} {\mathify{\widehat{g}}\xspace}
\newcommand{\hhat} {\mathify{\widehat{h}}\xspace}

\newcommand{\ICtilde} {\mathify{\widetilde{IC}}\xspace}

\newcommand{\Exp}{\mathop{\mathbb{E}}}
\newcommand{\half}{\frac{1}{2}}

\newcommand{\maj}{\mathop{\mathrm{majority}}}

\newcommand{\xor} {\mathop{\oplus}}

\newcommand{\Z}{{\mathbb Z}}

\def\noqed{\renewcommand{\qedsymbol}{}}

\newenvironment{IndentSection}{\par\begingroup\addtolength{\leftskip}{1em}}{\par\endgroup}
\fi
\ifx\THMTEX\undefined
\newcommand{\THMTEX}

\theoremstyle{plain}
\newtheorem{theorem}                {Theorem}[section]
\newtheorem{lemma}      [theorem]   {Lemma}
\newtheorem*{lemmaREF}              {Lemma}
\newtheorem*{propREF}               {Proposition}

\newenvironment{PropositionWithReference}[1]{\begin{propREF}[Proposition \ref{#1}]}{\end{propREF}}
\newtheorem{corollary}  [theorem]   {Corollary}
\newtheorem{proposition}[theorem]   {Proposition}

\newtheorem*{theoremUN}             {Theorem}
\newtheorem*{theoremREF}             {Theorem}

\newtheorem*{lemmaUN}               {Lemma}

\newtheorem*{propositionUN}         {Proposition}

\theoremstyle{definition}
\newtheorem{definition} [theorem]   {Definition}

\theoremstyle{remark}

\newenvironment{proof-of-claim}[0]{\begin{proof}[\bf Proof of Claim]}{\end{proof}}
\newenvironment{proof-of-lemma}[1]{\begin{proof}[\bf Proof of Lemma #1]}{\end{proof}}
\newenvironment{proof-of-proposition}[1]{\begin{proof}[\bf Proof of Proposition #1]}{\end{proof}}
\newenvironment{proof-sketch}[0]{\begin{proof}[\bf Proof sketch]}{\noqed\end{proof}}

\fi
\ifx\SCTEX\undefined
\newcommand{\SCTEX}{}

\newcommand{\TRUE}{\texttt{True}\xspace}
\newcommand{\FALSE}{\texttt{False}\xspace}

\newcommand{\OSI}[3][]{\ensuremath{P^{#1}_{#2}(#3)}\xspace}
\newcommand{\Inf}[3][]{\ensuremath{I^{#1}_{#2}(#3)}\xspace}
\newcommand{\FeasableAssignments}{\ensuremath{\mathbb{X}}\xspace}
\newcommand{\PAIR}[2]{\mathify{\left<#1,#2\right>}\xspace}
\newcommand{\TRIPLET}[3]{\mathify{\left<#1,#2,#3\right>}\xspace}

\newcommand{\Agenda}{\FeasableAssignments\xspace}
\newcommand{\mANDagenda}{\mathify{\left<A^1,\ldots,A^m,\mathop{\wedge}\limits_{j=1}^mA^j\right>}\xspace}
\newcommand{\kXORagenda}[1]{\mathify{\left<A^1,\ldots,A^{#1},\xor\limits_{j=1}^{#1}A^j\right>}\xspace}
\newcommand{\mXORagenda}{\kXORagenda{m-1}}
\newcommand{\LIN}{\texttt{Lin}\xspace}
\newcommand{\OLIG}{\texttt{Olig}\xspace}

\newcommand{\DISTx}[2]{d(#1,#2)}
\newcommand{\JUNTA}[1]{\mathify{#1_{\rightarrow J}}}
\fi

\def\monthname{%
    \ifcase\month
        \or Jan\or Feb\or Mar\or Apr\or May\or Jun%
        \or Jul\or Aug\or Sep\or Oct\or Nov\or Dec%
    \fi}%

\def\timestring{\begingroup
    \count0 = \time \divide\count0 by 60
    \count2 = \count0 
    \count4 = \time \multiply\count0 by 60
    \advance\count4 by -\count0 
    \ifnum\count4<10 \toks1 = {0}
    \else \toks1 = {}%
    \fi
    \ifnum\count2<12 \toks0 = {a.m.}%
    \else \toks0 = {p.m.}%
    \advance\count2 by -12 \fi
    \ifnum\count2=0 \count2 = 12 \fi 
    \number\count2:\the\toks1 \number\count4 \thinspace \the\toks0
\endgroup}%
\newcommand{\TITLE}{Approximate Judgement Aggregation
}
\newcommand{\AUTHOR}{Ilan Nehama}
\newcommand{\ADDRESS}{%
    Center for the Study of Rationality \\
    \quad\& The Selim and Rachel Benin School of Computer Science and Engineering\\
    The Hebrew University of Jerusalem, Israel}
\newcommand{\EMAIL}{ilan.nehama@mail.huji.ac.il}

\newcommand{\KEYWORDS}{%
    approximate aggregation,
    discursive dilemma,
    truth-functional agendas,
    inconsistency index,
    dependency index,
    computational social choice}

\newcommand{\ACK}{The research was supported by a grant from the Israeli Science Foundation (ISF)
    and by the Google Inter-university center for Electronic Markets and Auctions}
\title{\TITLE}
\author{\AUTHOR\\\ADDRESS\\\EMAIL}
\begin{document}
\maketitle
\begin{abstract}
In this paper we analyze judgement aggregation problems in which a
group of agents independently votes on a set of complex propositions
that has some interdependency constraint between them (e.g.,
transitivity when describing preferences).
We generalize the previous results by studying approximate judgement
aggregation. We relax the main two constraints assumed in the
current literature, Consistency and Independence and consider
mechanisms that only approximately satisfy these constraints, that
is, satisfy them up to a small portion of the inputs.
The main question we raise is whether the relaxation of these
notions significantly alters the class of satisfying aggregation
mechanisms.
The recent works for preference aggregation of Kalai, Mossel, and
Keller fit into this framework. The main result of this paper is
that, as in the case of preference aggregation, in the case of a
subclass of a natural class of aggregation problems termed
`truth-functional agendas', the set of satisfying aggregation
mechanisms does not extend non-trivially when relaxing the
constraints.
Our proof techniques involve boolean Fourier transform and analysis
 of voter influences for voting protocols.

The question we raise for Approximate Aggregation can be stated
using terms of Property  Testing. For instance, as a corollary from
our result we get a generalization of the classic result for
property testing of linearity of boolean functions.
\end{abstract}
\noindent{\bf Keywords:} \KEYWORDS
{\renewcommand{\thefootnote}{}\footnotetext{\ACK}}
{\renewcommand{\thefootnote}{}\footnotetext{Previous versions of this work were presented at
Bertinoro Workshop on Frontiers in Mechanism Design 2010,
Third International Workshop on Computational Social Choice, D$\ddot{u}$sseldorf 2010,
and Computation and Economics Seminar at the Hebrew University.
The author would like to thank the participants in these workshops for their comments.}}
\hypersetup{
    pdfsubject = {\TITLE},
    pdfkeywords = {\KEYWORDS},
    pdfauthor = {\AUTHOR},
    pdftitle = {\TITLE},
}

\section{Introduction}
\subsection{Abstract Judgement Aggregation}
Judgement Aggregation deals in scenarios in which one should aggregate a set of opinions/judgement done by independent judges or agents to one opinion.
For instance, assume a committee needs to decide whether to deploy a
suggested network protocol.
It is acceptable by all voters that the protocol should be deployed
if and only if it stands in two criteria: security (resistance to
attacks) and scalability (supports several sizes of a network).
We assume that each voter decides his opinion on the two criteria
independently and based on this decides on whether to deploy the
protocol. The voters cast their votes simultaneously and we assume
no strategic behavior on their behalf.
Now assume that a third of the voters are convinced that the
protocol stands in the two criteria and hence think it should be
deployed while the other voters think it should not be deployed but
disagree on the criterion that is violated. One third of the voters
think it is not secure enough and hence should not be deployed
although it is scalable. On the other hand, a third of the voters
think it is secure but not scalable enough and hence should not be
deployed.
Then we have that although a minority of the committee (a third of
the voters) thinks the protocol should be deployed, for each
criterion separately there is a supporting majority (two thirds)
that thinks the protocol passes. This discrepancy between the
majority vote on premises (security and scalability) and the
majority vote on the conclusion (deploy) was presented by Kornhauser
and Sager in 1986\cite{Kornhauser1986} and was later named `The
Doctrinal Paradox'.
The deploy-network scenario is an example for a case when the
decision is based on several arguments (while the logical connection
between the arguments and the decision is commonly accepted) and we
need to aggregate the decisions of (many) individuals into one
decision. In general, one can think of scenarios in which the
accepted opinions do not have this premise-conclusion structure ,for
instance the Arovian agenda of preferences.
A discrepancy phenomena as described above can happen for a lot of
such sets of `accepted opinions' (e.g., Condorcet Paradox for
preference aggregation) and is the subject of a growing body of
works in economics, political science, philosophy, law, and other
related disciplines. (A survey of this field can be found in
\cite{List2009} and \cite{List2008})
\newcommand{\FootnoteNonBooleanisPossibleToo}{\footnote{%
    There is some literature on aggregating non-binary issues, e.g.,
    \cite{Rubinstein1986}
    and \cite{Dokow2010a},
    but this is outside the scope of this paper.
}}%
Abstract aggregation is the field that deals with such problems that arise when aggregating several opinions (constrained to be in some given group, the agenda) to one opinion. It can be formalized in the following way. There
is a committee of $n$ individuals (also called voters) that needs to
decide on $m$ binary issues (that is, each question has exactly two
possible answers \TRUE and \FALSE\FootnoteNonBooleanisPossibleToo).
Each individual holds an {\bf opinion} which is an answer for each
of the issues. We denote the answer of the \ith voter for the \jth
issue by $X^j_i$ and the vector of all opinions in the committee
(called {\bf profile}) by
    \mbox{$X\in\left(\left\{0,1\right\}^m\right)^n$} (For the ease of presentation we will identify
\TRUE with 1 and \FALSE with 0).
Like in the example above, not all opinions are acceptable
    (one cannot accept a network protocol if it is not secure).
We assume a non-empty set \FeasableAssignments of $\{0,1\}^m$ called
the {\bf agenda } is given. The opinions in \FeasableAssignments are
called the {\bf consistent} opinions and only these opinions are
held by voters\footnote
    {For instance those might be the legal opinions, logic
    consistent opinions, or rational according to other criterion so one
    can assume that any `reasonable' individual should hold only
    consistent opinions.}.
For instance the {\bf conjunction agenda}, which is the agenda
described in the example, is defined to be the set
    \mbox{$\{000, 010, 100, 111\}$}\footnote{I.e., the third bit ia
a conjunction of the first two.}. 
A related model (that is more common in the literature) is
`Judgement Aggregation'. In this model the issues are logical
propositions over a set of variables and a consistent opinion is an
assignment to these variables (so not every combination of truth
values for the proposition is achievable). We feel that the logical
model is more suitable in cases where one searches for aggregation
mechanisms that respect the semantics of the agenda (which is
represented by the propositions). Since our work can be stated as
impossibility results, we prefer using the more general model. In
\cite{Dokow2009} Dokow and Holzman prove that the two models are
equivalent in the sense that each set of consistent opinions can be
described using a proposition set (although not uniquely).

An aggregation mechanism is a function that defines for any profile
the {\bf aggregated opinion}
(\mbox{$F:\left(\left\{0,1\right\}^m\right)^n\!\!\rightarrow\!\left\{0,1\right\}^m$}).
In this work we concentrate on two properties of aggregation
mechanisms, which were the first to be studied, independence and
consistency.
{\bf Independence} states that the aggregated opinion on the \jth
issue, $(F(X))^j$ depends solely on the opinions on that issue
$X^j$. For instance, issue-wise majority that was presented above
satisfies independence.
{\bf Consistency} of the aggregation mechanism states that whenever
all the members of the committee hold consistent opinions, i.e.,
\mbox{$X\in\FeasableAssignments^n$}, $F$ returns a consistent
opinion as well, i.e., \mbox{$F(X)\in\FeasableAssignments$}. For instance, for the conjunction agenda issue-wise majority
, as can be seen in the
    deploy-protocol
example, does not satisfy consistency. The consistency criterion is
commonly accepted as regarding a system of agents as an agent and
hence requiring it the same rationality constraint we require from
the individual agents. For instance, firms(in which decisions are
formed by
    aggregating opinions of different board members) act as a unit and
    it matters to the firms`s credibility if it can be expected to
    behave coherently.
The independence criterion, while not being uncontroversial, has
some appeal in that it guarantees a proposition-wise approach to
aggregation. In section \ref{Section_WhyIndependence} we discuss
further this criterion and claim that it can be justified in many
scenarios.

A natural question is to characterize the aggregation mechanisms that satisfy independence and consistency (for a given agenda).
For the conjunction agenda (under mild and natural constraint\footnote
    {Voter Sovereignty - For each issue, there a profile for which the aggregated opinion will be \FALSE and a profile for which the aggregated opinion will be \TRUE.})
    the only aggregation mechanisms that satisfy independence and consistency are the oligarchies (The oligarchy of a coalition $S$ returns for each issue \TRUE iff all voters in $S$ voted \TRUE for that issue).
Since, the set of consistent and independent aggregation mechanisms
is small and unnatural this characterization theorem is referenced
as an impossibility result.
Such impossibility results are quite strong.
They shows the impossibility of finding any aggregation mechanism that satisfies the conditions.
In other words: for each mechanism there will always be some judgement
profiles that lead to a breakdown of the mechanism.
In this work we deal with this point (and we think we even strengthening the impossibility result) by quantifying the number of profiles that lead to such breakdown.

For the conjunction agenda we prove:
\begin{theoremUN}\footnote{This is a special case of corollary
\ref{Main_Theorem_As_Corollary}.}
For any \mbox{$\epsilon>0$},
\mbox{$n\geqslant 1$},
    \mbox{$\delta<Cn^{-2}\epsilon^4$} (for some constant $C$ that
does not depend on $n$), no aggregation mechanism $F$ for the
conjunction agenda over $n$ voters satisfies the following five
conditions:
    \begin{itemize}\addtolength{\itemsep}{-0.5\baselineskip}
        \item $F$ is $\delta$-close to being independent.\\
            I.e., there exists an independent (not necessarily consistent) aggregation mechanism $G$ that returns the same aggregated opinion as $F$ for at least $(1-\delta)$ fraction of the profiles.
        \item $F$ is $\delta$-close to being consistent.\\
            I.e., $F$ returns a consistent result for at least $(1-\delta)$ fraction of the profiles.
        \item $F$ is $\epsilon$-far from returning always \FALSE on one of the basic premises.
        \begin{itemize}
        \item $F$ returns \TRUE as an aggregated opinion for the first issue for at least at least $\epsilon$ fraction of the profiles.
        \item $F$ returns \TRUE as an aggregated opinion for the second issue for at least at least $\epsilon$ fraction of the profiles.
        \end{itemize}
        \item $F$ is $\epsilon$-far from being an oligarchy.\\
            I.e., for any oligarchy $G$ of a coalition $S$, $F$ returns the same aggregated opinion as $G$ for at most $(1-\epsilon)$ fraction of the profiles.
    \end{itemize}
\end{theoremUN}
In this paper we prove similar theorems for a family of agendas:
truth-functional agendas in which every issue is either a premise or
a conclusion defined as conjunction or xor of several premises (up to input \& output
negation).
In a truth-functional agenda the issues are divided into two
types: premises and conclusions. Each conclusion $j$ is
characterized by a binary function $\Phi_j$ over the premises and
an opinion is consistent if the answers to the conclusion issues are
attained by applying the function $\Phi_j$ on the answers to the
premise issues.
\[\FeasableAssignments=
    \left\{\!\!\!\begin{array}{l|l}
    x\in\{0,1\}^m &
    x^j=\Phi^j(\text{premises})~~~\text{ for every conclusion issue $j$.}
    \end{array}\!\!\!\right\}\]
For instance the conjunction agenda is a truth-functional agenda
with two premises and one conclusion and we mark this by notating
the agenda as \TRIPLET{A}{B}{A\wedge B}. We restrict ourselves to
truth-functional agendas in which the conclusions are either
conjunction or xor up to negation of inputs and outputs\footnote
    {E.g., $A\wedge B$ and $\overline{A\wedge \overline{B}}$, $\overline{A \xor \overline{B} \xor C}$}.
\subsection{Approximate Judgement Aggregation}
Lately there is a series of works coping with impossibility results
in Social Choice using approximations (e.g.,
\cite{IoannisCaragiannis2010} and \cite{Friedgut2008}). The version
of approximation we define in this work is studying aggregation
mechanisms that are almost consistent and almost independent.
We quantify being almost consistent by defining {\bf $\boldsymbol\delta$-consistency} of an aggregation mechanism as having a consistent aggregation mechanism $G$ that disagrees with $F$ on at most $\delta$ fraction of the inputs\footnote
    {Formally,
    $\Pr\left[F(X)\neq G(X)~|~X\in\FeasableAssignments^n\right]\leqslant\delta$.}.
Similarly,  we quantify being almost independent by defining {\bf $\boldsymbol\delta$-independence} of an aggregation mechanism as having an independent aggregation mechanism that disagrees with $F$ on at most $\delta$ fraction of the inputs.
Both terms can be equivalently defined as the failure probability of tests as we show in section \ref{Section_Model} and actually the second versions of the definitions are those we use in the analysis since they are easier to work with.

Both definitions use the Hamming distance
    \mbox{$d^\Agenda(F,G)=\Pr\left[F(X)\neq G(X)~|~X\in\FeasableAssignments^n\right]$}
which includes an assumption of uniform distribution over the opinions for each voter and that voters draw their opinions independently ({\bf Impartial Culture Assumption}). This assumption, while certainly unrealistic, is the natural choice in this kind of work and is discussed further in \mbox{section \ref{Section_Model}.}

Trivially, if $F$ is close to $G$ and $G$ is independent and
consistent, $F$ is $\delta$-independent and $\delta$-consistent for
$\delta$ linear in $d(F,G)$. Our main question is whether there are
aggregation mechanisms that are close to being independent and
consistent that are not close the (usually small) set of consistent
and independent aggregation mechanisms.

In several agendas the set of consistent independent aggregation
mechanisms is a very small set(E.g., dictatorship or oligarchies)
and hence this question is equivalent to asking whether we can look for
aggregation mechanisms that are close to being independent and close to being
consistent without collapsing to the known small (perturbed) set of consistent independent aggregation mechanisms.
\subsection{Connection to Approximate Preference Aggregation}
Preceded our work is the works by Kalai\cite{Kalai2002}, Mossel\cite{Mossel2010}, and
Keller\cite{Keller2010} that proved similar approximate aggregation theorems for {\bf preference aggregation}.
In this agenda the consistent opinions represent
the linear orders over a set of candidates
\mbox{$\{c_1,c_2,\ldots,c_s\}$} and the issues are the
$\binom{s}{2}$ pair-wise comparisons between candidates\footnote
    {\newcommand{\IsGt}{{\scriptscriptstyle\overset{?}{\geqslant}}}%
For instance, for \mbox{$s\!=\!3$} (three candidates) the issues are \mbox{`$c_1\IsGt c_2$'},
    \mbox{`$c_2\IsGt c_3$'}, and \mbox{`$c_3\IsGt c_1$'} and the consistent opinions are \mbox{$\left\{001,\!010,\!100,\!110,\!101,\!011\!\right\}$}.}.%
Similarly to the doctrinal paradox, the Condorcet
Paradox\cite{Condorcet1785} shows that issue-wise majority might lead to an inconsistent result for preference aggregation as well. Arrow's theorem\cite{Arrow1950} shows that (under mild and natural constraint\footnote
    {Pareto - Whenever all the voters hold the same opinion, this is the aggregated opinion.
    }) the only aggregation mechanisms that satisfy independence and consistency are the dictatorships.

The recent works of Kalai\cite{Kalai2002}, Mossel\cite{Mossel2010}, and
Keller\cite{Keller2010} proved similar results to the results we prove here.
\begin{theoremUN}[\cite{Keller2010}]
There exists an absolute constant $C$ such that the following holds:
For any \mbox{$\epsilon>0$} and \mbox{$k\geqslant3$},
    if $f$ is an aggregation mechanism for the preference agenda
    over $k$ candidates that satisfies independence and
    $C \cdot \left(\epsilon/k^2 \right)^3$-consistency,
    , then there exists an aggregation mechanism $G$ that satisfies
    independence and consistency such that \mbox{$d(F,G)<\epsilon$}.
\end{theoremUN}
Using the technique we show in this paper (theorem \ref{Thm_RelaxingIIA}), one can generalize this result for characterizing the $\delta$-independent $\delta$-consistent aggregation mechanisms for preference mechanisms.
\subsection{Connection to Property Testing}
We think it is useful to phrase the question of approximate
aggregation using terms of property testing. In this field we query
a function at a small number of (random) points testing for a global
property (In our case, the property is being a consistent independent
aggregation mechanism). We discuss this connection further is
section \ref{Section_PropertyTesting}. For example a corollary of
the results we present in this paper (in property testing terms):
\begin{IndentSection}
\noindent For any three binary functions
\mbox{$f,g,h:\{0,1\}^n\rightarrow\{0,1\}$}, if
\[\Pr[f(x)\xor g(y) = h(x\xor y)]\geqslant 1-\epsilon\]
(when the addition is in $\Z_2$ and $\Z_2^n$, respectively), then
there exists three binary functions
\mbox{$f',g',h':\{0,1\}^n\!\rightarrow\!\{0,1\}$} such that
    \mbox{$\Pr[f(x)\!\neq\! f'(x)]$},
    \mbox{$\Pr[g(x)\!\neq\! g'(x)]$}, and
    \mbox{$\Pr[h(x)\!\neq\! h'(x)]$} are smaller than $C\epsilon$
for some constant $C$ independent of $n$ and
\[\forall x,y~:~f'(x)\xor g'(y) = h'(x\xor y).\]
\end{IndentSection}
Notice a special case of this result (for $f=g=h$) is the classic
result of Blum, Luby, and Rubinfeld (\cite{Blum1993},
\cite{Bellare1995}) for linear testing of boolean functions.
\subsection{Techniques}
We prove the main theorem by proving the specific case of
independent aggregation mechanism for two basic agenda families: the
conjunction agendas (agendas in which there is exactly one conclusion with is constrained to be the conjunction of the premises. \ref{Main_Theorem_mAND})
and the xor agendas (agendas in which there is exactly one conclusion with is constrained to be the xor of the premises. theorem
\ref{Main_Theorem_mXOR}). Later we show how to extend these
theorems to the general theorem of relaxing both constraints (theorem \ref{Thm_RelaxingIIA}).

We use two different techniques in the proofs.
For the conjunction agendas we study influence measures\footnote{Both
the known influence (Banzhaf power index) and a new measure we
define: The ignorability of a voter.} of voters on the
issue-aggregating functions
and for the xor agendas we use Fourier analysis of the
issue-aggregating functions.

An open question is whether one can find such bounds for any agenda
or whether there exists an agenda for which the class of aggregation
mechanisms that satisfy consistency and independence expands non
trivially when we relax the consistency and independence
constraints.

We proceed to describe the structure of the paper.
In Section 2 we describe the formal model of aggregation mechanisms.
In section 3 we give present the main
agendas we deal with, truth-functional agendas, and specifically
conjunction agendas and xor agendas.
In section 4 we state the motivation to deal with approximate
aggregation, and describe the known results for preference
approximate aggregation by Kalai, Mossel, and Keller.
In section 5 we describe the connection we find between Approximate
Aggregation and the field of Property Testing.
In sections 6 and 7 we describe our main theorems and outline the
proof.
Section 8 concludes.
\section{The Model}
\label{Section_Model}%
We define the model similarly to
\cite{Dokow2009} (which is Rubinstein and Fishburn's model
\cite{Rubinstein1986} for the binary case)

We consider a {\bf committee} of $n$ individuals that needs to
decide on $m$ issues. An {\bf opinion} is a vector
\mbox{$x=(x_1,x_2,\ldots,x_m)\in\{0,1\}^m$} denoting an answer to
each of the issues. An opinion {\bf profile} is a matrix
\mbox{$X\in\left(\left\{0,1\right\}^m\right)^n$} denoting the
opinions of the committee members so
    an entry $X^j_i$ denotes the vote of the \ith voter for the \jth issue,
    the \ith row of it $X_i$ states the votes of the \ith individual on all issues,
    and the \jth column of it $X^j$ states the votes of each of the individuals on
        the \jth issue.
In addition we assume that an {\bf agenda}
\mbox{$\FeasableAssignments\in\{0,1\}^m$} of the {\bf consistent}
opinions is given.

The basic notion in this field is an {\bf aggregation mechanism}
which is a function that returns an {\bf aggregated opinion} (not
necessarily consistent) for every profile
    \mbox{$\left(F:\left(\left\{0,1\right\}^m\right)^n\rightarrow\{0,1\}^m\right)$}\footnote
    {We define the function for all profiles for simplicity but we are not
        interested in the aggregated opinion in cases one of the voters
        voted an inconsistent opinion.}.

An aggregation mechanism satisfies {\bf Independence} (and we say
that the mechanism is {\bf independent}) if for any two consistent
profiles $X$ and $Y$ and an issue $j$, if \mbox{$X^j=Y^j$} (all
individuals voted the same on the \jth issue in both profiles) then
\mbox{$(F(X))^j=(F(Y))^j$} (the aggregated opinion for the \jth
issue is the same for both profiles). This means that $F$ satisfies
independence if one can find $m$ binary functions
\mbox{$f^1,f^2,\ldots,f^m:\{0,1\}^n\rightarrow\{0,1\}$} s.t.
\mbox{$F(X)\equiv\left(f^1(X^1),f^2(X^2),\ldots,f^m(X^m)\right)$}%
\footnote
    {Notice this property is a generalization of the IIA property
    for social welfare functions (aggregation mechanism for the
    preference agenda) so a social welfare function satisfies IIA iff it
    satisfies independence as defined here (when the issues are the
    pair-wise comparisons).}.
An independent aggregation \mbox{mechanism} satisfies {\bf
systematicity} if \mbox{$F(X)=\left<f(X^1),\ldots,f(X^m)\right>$}
for some issue aggregating function $f$, i.e., all issues are aggregated
using the same function.
We will use the notation \mbox{$\left<f^1,f^2,\ldots,f^m\right>$}
for the independent aggregation mechanism that aggregates the \jth
issue using $f^j$.

The main two measures we study in this paper are the {\bf
inconsistency index} $IC^\Agenda(F)$ and the {\bf dependence index}
$DI^\Agenda(F)$ of a given aggregation mechanism $F$ and a given
agenda \FeasableAssignments. These
measures are relaxations of the {\bf consistency} and {\bf
independence} criterion that are usually assumed in current
works\footnote
    {$F$ satisfies consistency iff \mbox{$IC(F)=0$} and independence iff \mbox{$DI(F)=0$}}.
We define the measures by
\begin{definition}[Inconsistency Index]~\\
For an agenda \Agenda and an aggregation mechanism $F$ for that
agenda, the {\bf inconsistency index} is defined to be the
probability to get an inconsistent result.
    \[IC^\Agenda(F)=\Pr\left[F(X)\notin\FeasableAssignments~|~X\in\FeasableAssignments^n\right].\]
\end{definition}
\begin{definition}[dependency index]~\\
For an agenda \Agenda and an aggregation mechanism $F$ for that
agenda, the {\bf dependency vector} $DI^{j,\Agenda}(F)$ is defined
as
{\newcommand{\EXPR}{\Exp\limits_{X\in\FeasableAssignments^n}\left[
        \Pr\limits_{Y\in\FeasableAssignments^n}\left[(F(X))^j\neq (F(Y))^j | X^j=Y^j\right]
        \right]}
    \[DI^{j,\Agenda}(F)=\EXPR.\]
}%
    {The definition can be seen as a test for independence of the \jth issue as discussed
    in section \ref{Section_PropertyTesting}}\\
The {\bf dependency index} $DI^\Agenda(F)$ is defined by
    \[DI^\Agenda(F)=\max\limits_{j=1,\ldots,m}DI^{j,\Agenda}(F)\]
\end{definition}

In contexts where the agenda is clear we omit the agenda superscript and notate
these as \mbox{$IC(F)$}, \mbox{$DI^j(F)$}, and \mbox{$DI(F)$}, respectively.

We define these two indexes using local tests. We prove that the more natural definition of distance to he class of aggregation mechanisms that satisfy consistence (or
independence) is equivalent to the above (up to multiplication by $\delta$ by constant).

\begin{proposition}
Let $F$ be an aggregation mechanism for an agenda over $m$ issues.
Then $F$ satisfies \mbox{$IC(F)\leqslant\delta$}
iff there exists a consistent aggregation mechanism $H$ that
satisfies \mbox{$d(F,H)\leqslant\delta$}.
\end{proposition}

\begin{proposition}\label{Proposition_DIJisEpsIndependent}
Let $F$ be an aggregation mechanism and $j$ an issue.
If \mbox{$DI^j(F)\leqslant\epsilon$}, then there exists an
aggregation mechanism $H$ that satisfies $DI^j(H)=0$ and
\mbox{$d(F,H)\leqslant 2\epsilon$}.
If \mbox{$DI^j(F)\geqslant\epsilon$}, then every aggregation
mechanism $H$ that satisfies $DI^j(H)=0$, also satisfies
\mbox{$d(F,H)\geqslant \half\epsilon$}\\
\end{proposition}

\begin{proposition}\label{Proposition_DIisEpsIndependent}
Let $F$ be an aggregation mechanism for an agenda over $m$ issues
that satisfies \mbox{$DI(F)\leqslant\delta$}.
Then there exists an independent aggregation mechanism $H$ that
satisfies \mbox{$d(F,H)\leqslant 2m\delta$}.
\end{proposition}

These definitions include two major assumptions on the opinion
profile distribution. First, we assume the voters pick their
opinions independently and from the same distribution. Second, we
assume a uniform distribution over the (consistent) opinions for
each voter ({\bf Impartial Culture Assumption}).
The uniform distribution assumption, while certainly unrealistic, is
the natural choice for proving `lower bounds' on $IC(F)$. That is,
proving results of the format `Every aggregation mechanism of a
given class has inconsistency index of at least $\gamma(n)$'. In
particular, the lower bound, up to a factor $\delta$, applies also
to any distribution that gives each preference profile at least a
$\delta$ fraction of the probability given by the uniform
distribution.\footnote
    {In successive works we relax this assumption and prove similar
    results for more general distributions.}%
 Note that we cannot hope
to get a reasonable bound result for every distribution. For
instance, since for every aggregation mechanism we can take a
distribution on profiles for which it returns a consistent opinion.
\subsection{The Independence Property}
\label{Section_WhyIndependence}%
The independence criterion is sometimes criticized as being
unjustified normatively in most real-life scenarios\footnote
        {Chapman(\cite{Chapman}) and Mongin(\cite{Mongin2008})
attack this criterion and claim it removes the discipline of
reason from social choice since it disregards the intra-issue
dependencies which is the essence of the problem. According to this
criterion the aggregation of `complex' issues is done without
regarding the reasons of the voters for their opinions and hence
lacks the information for good aggregation.}.
The impossibility results of judgement aggregation can also be seen
as `empirical' argument against independence since they show that it
contradicts the more natural assumption of consistency. While we
accept this argument, we think our  work quantifies the
tradeoff between the two criteria.
Moreover, in this section we claim that in the general case this criterion
can be justified on several different grounds.

First, in a lot of cases it is justified to expect due to normative
or legal reasons that changing an individual judgement on an issue
should not change the collective judgement on another issue.

Secondly, there are works that defend this criterion by using
manipulation-resistance arguments.
In \cite{Dietrich2007b} Dietrich
and List define the notion of manipulability of an aggregation
mechanism\footnote
    {An aggregation mechanism $F$ is manipulable at the profile $X$
    by individual $i$(the manipulator) on issue $j$ if
        \mbox{$X^i_j\neq\left(F\left(X\right)\right)_j$} (The
    manipulator disagrees with the aggregated opinion on issue $j$), but
        \mbox{$X^i_j=\left(F\left(X'\right)\right)_j$}
    for some profile $X'$ that differs form $X$ in $i$`s vote only.(I.e.
    the manipulator can get his will on $j$ by voting differently)%
    }
and proved that any aggregation mechanism that does not satisfy
independence is manipulable.
In this paper they further prove that this manipulability property
is equivalent to a more game-theoretic property of
strategy-proofness under some assumptions on players' preferences.

On the ground of simplicity of representation one can justify
independence as a criterion that returns aggregation mechanisms that
are easy to represent, calculate, or justify (For instance, justify
a voting method to the public).

Other grounds of justification for such aggregation mechanisms are
from the voter point of view.
There might be a situations in which the decisions are made over
time (different meetings) or place (different representatives of the
same voting identity) so it is fair to assume that when voting on an
issue or aggregating the votes it is unreasonable to depend on
votes on other issues.
Another argument might be than there are scenarios in which you need
to define the aggregation method and only at a later stage choose from the set of issues
the relevant one (For instance, the definition of Social Welfare
Functions as returning a choice function so only at a later stage
the society is faced with the menu of alternatives).
\subsection{Binary Functions}
Since this work deals with binary functions (for aggregating
issues), we need to define several notions for this framework as
well.
To ease the presentation, throughout this paper we will identify
\TRUE with 1 and \FALSE with 0 and use logical operators on bits and
bit vectors (using entry-wise semantics).

Let \mbox{$f:\{0,1\}^n\rightarrow\{0,1\}$} be a binary function.
$f$ is the {\bf oligarchy} of a coalition $S$ if it is of the form:
\mbox{$f(x)=\prod\limits_{i\in S}x_i$}. This means that $f$ returns $1$ if all
the members of $S$ voted 1. We denote by \pmb{\OLIG} the class of
all $2^n$ oligarchies. Two special cases of oligarchies are the
constant 1 function which is the oligarchy of the empty coalition
and the dictatorships which are oligarchies of a single voter.

$f$ is a {\bf linear} function if it is of the form
\mbox{$f(x)=\xor\limits_{i\in S}x_i$} for some coalition
$S$\footnote{\mbox{An equivalent definition is:
    \mbox{$\forall x,y:~f(x)\!+\!f(y)=f(x\!+\!y)$}} when the addition is in
$\Z_2$ and $\Z_2^n$, respectively.}. This means that $f$ returns $1$ if an even
number of the members of $S$ voted 1. We denote by \pmb{\LIN} the
class of all $2^n$ linear functions. Two special cases of linear
functions are the constant 1 function which is the xor function over
the empty coalition and the dictatorships which are xor of a single
voter.

We say that $f$ satisfies the {\bf Pareto} criterion if
\mbox{$f(\bar{0})=0$} and \mbox{$f(\bar{1})=1$}\footnote
    {In the literature this criterion is sometimes referred to as
    Unanimity, e.g., in \cite{List2009}. We choose to follow
    \cite{Dokow2010} and refer to it as Pareto to distinguish
    between it and the unanimity function which is the oligarchy
    of \mbox{$\{1,2,\ldots,n\}$}.}.
I.e., when all the individuals voted unanimously 0 then $f$
should return 0 and similarly for the case of 1.

We define two measures for the influence of an individual on a
function \mbox{$f:\{0,1\}^n\rightarrow\{0,1\}$}. Both definitions
use the uniform distribution over \mbox{$\{0,1\}^n$} (which is
consistent with the assumption we have on the profile distribution).
\begin{itemize}
\item The {\bf influence}\footnote{In the simple cooperative  games
regime, this is also called the Banzhaf power index of player $i$ in
the game $f$.} of a voter $i$ on $f$ is defined to be the
probability that he can flip the result by changing his vote.
    \[\Inf{i}{f}=\Pr[f(x)\neq f(x \xor e_i)]\]
($x \xor e_i$ : $e_i$ = the \ith elementary vector. It is equivalent
to flipping the \ith bit \mbox{$0\leftrightarrow1$})
\item The (zero-){\bf ignorability} of a voter $i$ on $f$ is defined to be
    the probability that $f$ returns 1 when $i$ voted 0.%
    \[\OSI{i}{f}=\Pr[f(x)=1~|~x_i=0]\]
(We did not find a similar index defined in the voting literature or
in the cooperative games literature).
\end{itemize}

In addition we define a distance function over the binary
functions. The distance between two functions
\mbox{$f,g:\{0,1\}^n\rightarrow\{0,1\}$} is defined to be the
probability of getting a different result (normalized Hamming
distance). \mbox{$d(f,g)=\Pr\left[f(x)\neq g(x)\right]$}. From
this measure we will derive a distance from a function to a set of
functions by \mbox{$d(f,\mathcal{G})=\min\limits_{g\in\mathcal{G}}d(f,g)$}

One more notation we are using in this paper is $x_{_J}$ for a binary vector
\mbox{$x\in\{0,1\}^n$} and a coalition \mbox{$J\subseteq\{1,2,\dots
n\}$} for notating the entries of $x$ that correspond to $J$.
\section{Agenda Examples}
A lot of natural problems can be formulated in the framework of
aggregation mechanisms. It is natural to divide the agendas to two
major classes {\bf Truth-Functional Agendas} and {\bf Non
Truth-Functional Agendas}.

\subsection{Truth-Functional Agendas}%
In a truth-functional agenda the issues are divided into two types:
$k$ premises and \mbox{$(m-k)$} conclusions. The conclusion issues
are binary functions over the $k$ premises,
\mbox{$\Phi:\{0,1\}^k\rightarrow\{0,1\}^{m-k}$}. An opinion is
consistent if the answers to the conclusion issues are attained by
applying the function $\Phi$ on the premise issues.
\[\FeasableAssignments=\left\{\begin{array}{l|r}
    x\in\{0,1\}^m&
    x^j=\Phi^j(x_1,\ldots,x_k)~~~j=k+1,\ldots,m
\end{array}\right\}\]
There are cases in which there might be more than one way to
classify the issues to premises and conclusions. For instance,
the 2-premises xor agenda described below
\mbox{$\FeasableAssignments=\left\{001,010,100,111\right\}$} can be
defined both as \TRIPLET{A}{B}{A\xor B} and as \TRIPLET{A}{A\xor C}{C}. Since
we choose to analyze the agenda as opinion sets (and not as a
proposition set) we do not handle this point and notice that it is
irrelevant for our results.

These agendas ,due their structure, seems to be a good point to
start our work on approximate aggregation and in this paper we prove basic results
for the following two families of truth-functional agendas. Later we
derive results to a more general family of truth-functional agendas.
\subsubsection{Conjunction Agendas}
In the ($m$-premises) conjunction agenda
\mANDagenda there are
\mbox{$m+1$} issues to decide on and the consistency criterion is
defined to be that the last issue is a conjunction of the other
issues. For instance the Doctrinal Paradox agenda is the 2-premises
conjunction agenda. A common description of the problem is of a
group of judges or jurors that should decide whether a defendant is
liable under a charge of breach of contract. Each of them should
decide on three issues: whether the contract was valid ($p$),
whether there was a breach ($q$) and whether the defendant is liable
($r$). In their decision making they are constrained by the legal
doctrine that the defendant is only liable if the contract was valid
and if there was indeed a breach (\mbox{$r \iff (p \wedge q)$}).
\subsubsection{Xor Agendas}
Similarly, in the ($m$-premises) xor agenda
\kXORagenda{m} there are
\mbox{$m+1$} issues to decide on and the consistency criterion is
defined to be that the third issue is \TRUE if the number of
true-valued opinions for the first $m$ is even. An equivalent way to
define this agenda is constraining the number of \TRUE answers to be
odd.

\subsection{Non Truth-Functional Agendas}%
One can think on a lot of agendas that cannot be represented as a truth-functional agenda.
Among such interesting natural agendas that were studied one can find the equivalence
agenda\cite{Fishburn1986}, the membership agenda
\cite{Rubinstein1998}\cite{Miller2008}, and the preference agenda described below.
\subsection{Preference Aggregation}
Aggregation of preferences is one of the oldest aggregation
frameworks studied. In this framework there are $s$ candidates and
each individual holds a full strict order over them. We are
interested in Social Welfare Functions which are functions that
aggregate $n$ such orders to an aggregated order. As seen in
\cite{Nehring2003} and \cite{Dietrich2007a}, this problem can be
stated naturally in our framework by defining $\binom{s}{2}$
issues\footnote{The issue $\left<i,j\right>$ (for \mbox{$i\!<\!j$}) represents
whether an individual prefers $c_i$ over $c_j$.}.
\section{Motivation \& Known Results}
We find the motivation for dealing with the field of approximate
aggregation in three different disciplines.
\begin{itemize}
\item The consistent characterization are often regarded as
    `impossibility results' in the sense that they `permit' a very
    restrictive set of aggregation mechanisms. (e.g., Arrow's theorem
    tells us that there is no `reasonable' way to aggregate
    preferences). Extending these theorems to approximate aggregation
    characterizations sheds light on these impossibility results by
    relaxing the constraints.
\item
    The questions of Aggregation Theory have often  roots in
    Philosophy, Law, and Political Science.
    There is a long line of works suggesting consistent aggregating
    mechanisms while still trying to stay `reasonably close' to
    independence. The main general (not agenda-tailored) suggestions
    are premise-based mechanisms and conclusion-based aggregation
    for truth-functional agendas(see, among others,
        \cite{Kornhauser1986}   \cite{Kornhauser1992}
        \cite{Pettit2001}       \cite{List2002}
        \cite{Dietrich2006}     \cite{Bovens2006}), and a
    generalization of them to non-truth-functional agendas called
    sequential priority aggregation(\cite{List2004a},
    \cite{Dietrich2007c}). Another procedure in the literature is the
    distance-based aggregation(\cite{Pigozzi2006}) which is a well
    known for preference aggregation (E.g., Kemeny voting
    rule\cite{Kemeny1959}, Dodgson voting rule\cite{Black1986}, and
    lately a more systematic analysis in \cite{Elkind}).
    Our work contribute to this discussion by pointing out where one
    should search for solutions while not leaving the consistency
    and independence constraints entirely.
\item Connections to Property Testing as discussed in section
\ref{Section_PropertyTesting}
\end{itemize}
The first work studying approximate aggregation was done for the
preference agenda over three candidate by Kalai\cite{Kalai2002}
(although without stating the general framework of approximate
aggregation). This work was generalized by Mossel\cite{Mossel2010},
and Keller\cite{Keller2010}.
\begin{theoremUN}[\cite{Keller2010} theorem 1.3]
There exists an absolute constant $C$ such that the following holds:
For any \mbox{$\epsilon>0$} and \mbox{$k\geqslant3$},
    if $f$ is an aggregation mechanism for the preference agenda
    over $k$ candidates that satisfies independence and
    $C \cdot \left(\epsilon/k^2 \right)^3$-consistency,
    , then there exists an aggregation mechanism $G$ that satisfies
    independence and consistency such that \mbox{$d(F,G)<\epsilon$}.
\end{theoremUN}
\section{Connection to Property Testing}
\label{Section_PropertyTesting}
In the words of \cite{Ron2001},
the field of property testing deals with the following:
\begin{IndentSection}
    Given the ability to perform (local) queries concerning a
    particular object (e.g., a function or a graph), the task is to
    determine whether the object has a predetermined (global)
    property (e.g., linearity or bipartiteness), or is far from
    having the property. The task should be performed by inspecting
    only a small (possibly randomly selected) part of the whole
    object, where a small probability of failure is allowed.\\
    Property testing trades accuracy (the distance parameter) for
    efficiency (number of queries).
\end{IndentSection}
We think it is useful to view the Approximate Aggregation problem in
the framework of Property Testing.
Below we highlight the connection between Approximate Aggregation
and a special case of Property testing termed `one-sided
non-adaptive program testing'. For a general survey of the field,
one can read \cite{Fischer2001a}, \cite{Ron2001}, and
\cite{Goldreich1999}.

In our case the global property we are trying to test is
`consistency and independence' of an aggregation mechanism. The class
of satisfying aggregation mechanism is characterized by the current state of
research.
It is clear that each of the components of this property separately,
consistency and independence of an issue, can be tested trivially.
The consistency test consists of picking a (consistent) profile
uniformly at random and checking whether the aggregated opinion is
consistent. The test for independence of issue $j$ consists of
picking a (consistent) profile uniformly at random altering the
opinion for each voter without changing the \jth bit and check
whether the aggregated opinion on the \jth issue is changed due to
the altering. For each of the two tests the probability to accept a
non-satisfying mechanism is linear in the distance to the satisfying
set (and equals $IC(F)$ and $DI^j(F)$, respectively).
The main question of this work (in property testing terms) can be
stated as follows: What is the best test for being `consistent and
independent' one can assemble from running the ($m$+1) tests as
black boxes (and therefore get information only on $IC(F)$ and
$DI^j(F)$).
Similar question was asked lately in \cite{Chen2010}. In
\cite{Chen2010} the authors query (among other similar questions)
the conditions needed in order to deduce from testability of two
properties the testability of the intersection of the two properties.
Our work can be seen as studying this question for a specific domain
in which the question seems to be natural and while adding the
constraint that the test of the intersection property should be
defined as a sequence of tests for the basic properties.

The main result of this paper is that for a class of mechanisms
(corresponding to a natural class of agendas) one can assemble those
tests to a test for the property `consistent and independent'.

Similarly one can state questions dealing with sub-families of
aggregation mechanisms. For example, as we stated in the
introduction, the classic result of Blum, Luby, and Rubinfeld for
linearity testing of boolean functions is a direct corollary of our
result for the 2-premises xor agenda when considering systematic aggregation
mechanisms.

Still, the target of the two fields is different. While Property
Testing deals with finding the most efficient (query-wise) algorithm
for testing a property (functions family), Approximate Aggregation
deals with analyzing a specific family of tests.
\section{Main  Results}
\label{Section_Results}%
The main result of this paper is
\begin{theorem}\label{MainTheorem_IC_DC}~\\
For any \mbox{$\epsilon>0$} and \mbox{$m,n\geqslant 1$}, there
exists \mbox{$\delta_{_{IC}},\delta_{_{DI}}
    =n^{-2}\left(\frac{\epsilon}{m}\right)^{poly(m)}$}, such that
for every truth-functional agenda \Agenda over $m$ issues, in
which each non-premise issue is defined to be either
        conjunction of several premises
    or
        xor of several premises%
(up to negation of inputs or output)\footnote
    {
        For example
        \mbox{\TRIPLET{A}{B}{A\wedge B}},
        \mbox{$\TRIPLET{A}{B}{A\rightarrow B} \equiv \TRIPLET{A}{B}{\overline{A \wedge \overline{B}}}$},
        \mbox{$\left<A,B,C,A \wedge B, B \xor C, A \vee C\right>$}.%
    },
if $F$ is an aggregation mechanism for \Agenda over $n$ voters
satisfying
    \mbox{$\delta$-independence} and \mbox{$\delta$-consistency},
then there exists an aggregation mechanism $G$ that satisfies
consistency and independence such that
\mbox{$d(F,G)<\epsilon$}\\Moreover, one can take
{%
\newcommand{\EXPRA}{
    n^{-2}\left(\frac{\left(1-\beta_\epsilon\right)\epsilon}{5m}\right)^{m^2\!+\!m\!-1}
    \!\!-\!\beta_\epsilon\epsilon}%
\newcommand{\EXPRB}{\frac{1}{2m}\beta_\epsilon\epsilon}%
    \mbox{$\delta_{_{IC}}=\EXPRA$}
    ~and~
    \mbox{$\delta_{_{DI}}=\EXPRB$}
}%
 for any \mbox{$\beta_\epsilon\in\left[0,n^{-2}\left(\frac{\epsilon}{m}\right)^{m^2\!+\!m\!-1}\right]$}.
\end{theorem}
A direct corollary is the following impossibility result.
\begin{corollary}\label{Main_Theorem_As_Corollary}
For any \mbox{$m,n\geqslant 1$} and \mbox{$\epsilon,\delta\in[0,1]$} s.t.
    \mbox{$\delta<n^{-2}\left(\frac{\epsilon}{10m}\right)^{m^2+m-1}$},
,and a truth-functional agenda \Agenda over $m$ issues,
    in which each non-premise issue is defined to be either
        conjunction of several premises
    or
        xor of several premises%
(up to negation of inputs or output)%
, no aggregation mechanism $F$ for \Agenda over $n$ voters satisfies
the following three conditions:
    \begin{itemize}\addtolength{\itemsep}{-0.5\baselineskip}
        \item \mbox{$\delta$-independence}
        \item \mbox{$\delta$-consistency}
        \item $F$ is $\epsilon$-far from any independent and consistent aggregation mechanism for \Agenda.
    \end{itemize}
\end{corollary}

In the case of xor agenda (and its generalization, a
truth-functional agenda in which all the conclusions are xor)
we can get a much better result (e.g., no dependence on the number
of voters)
\begin{theorem}\label{Main_Theorem_mXOR_IC_DC}
Let $m\geqslant 3$ and let the agenda be
    \mbox{$\Agenda=\left<A^1,\ldots,A^{m-1},\xor\limits_{j=1}^{m-1}A^j\right>$}.
For any \mbox{$\epsilon<\frac{1}{6}$} and any aggregation mechanism $F$:\\
If $F$ is an aggregation mechanism for \Agenda over $n$ voters satisfying
\mbox{$\epsilon$-independence} and \mbox{$\epsilon$-consistency}, then
there exists an aggregation mechanism $G$ that satisfies consistency
and independence such that \mbox{$d(F,G)<m(2m+3) \epsilon$}%
\end{theorem}
Noticing that any affine agenda (i.e., an agenda that is an affine
subspace) can be represented as a truth-functional agenda that uses
xor conclusions only (lemma \ref{Lemma_AffineIsXOR}) we can get the following corollary
\begin{corollary}\label{Corollary_AffineSubspace}
For any \mbox{$\epsilon>0$} and \mbox{$m,n\geqslant 1$}, there
exists
    \mbox{$\delta=\frac{\epsilon}{m(2m+3)}$},
such that for every affine agenda \Agenda over $m$ issues, if $F$ is
an aggregation mechanism for \Agenda over $n$ voters satisfying
\mbox{$\delta$-independence} and \mbox{$\delta$-consistency}, then
there exists an aggregation mechanism $G$ that satisfies consistency
and independence such that \mbox{$d(F,G)<\epsilon$}%
\end{corollary}
\section{Proof Sketch of the Main Theorem}
In this section we sketch the techniques behind our proofs. The full
proofs can be found in the appendices.

We prove the main theorem by proving three independent theorems.
An approximation result for independent aggregation mechanisms for
conjunction agendas (theorem \ref{Main_Theorem_mAND}).
An approximation result for independent aggregation mechanisms for
xor agendas (theorem \ref{Main_Theorem_mXOR}).
An agenda independent method of converting results for the
independent case to the general case of relaxing both constraints
(theorem \ref{Thm_RelaxingIIA}).
Using induction on the number of conclusions and noticing that
negating (of the inputs and of the output) is renaming of opinions
in our framework (and hence does not change the approximation
results) we get theorem \ref{MainTheorem_IC_DC}.

\subsection{Conjunction Agenda}
For the agenda \mANDagenda we prove:
\begin{theorem}\label{Main_Theorem_mAND}
Let $m\geqslant 2$ and let the agenda be
    $\Agenda=\mANDagenda$.

For any \mbox{$\epsilon>0$} and any independent aggregation mechanism $F$:\\
    If \mbox{$IC(F)\leqslant\epsilon$}, then there exists an
    aggregation mechanism $G$ that satisfies consistency and
    independence such that
\mbox{$d(F,G)<5m \left(n^2\epsilon\right)^\frac{1}{m^2+m-1}$}.
\end{theorem}
{\bf Technique}: The main insight in the proof is that for any two
different issue-aggregating functions, $f$ and $g$, in $F$ for two
of the $m$ premises, we can bound the product of the influence
of a voter on $f$ and the ignorability of the same voter for $g$
using the inconsistency index of $F$ by
    \mbox{$\OSI{i}{f}\cdot\Inf{i}{g}\leqslant 4IC(F)$}.
\begin{proof-sketch}~\\
\indent Let $F=\TRIPLET{f}{g}{h}$ be an aggregation mechanism that satisfies
    \mbox{$IC(F)\leqslant\epsilon$}.
In case that $f$ (or $g$) is close enough to the constant zero
function, $F$ is close
    to the consistent aggregation mechanism \TRIPLET{0}{g}{0}.
\\\indent Otherwise, we define for
{a given function \mbox{$f:\{0,1\}^n\rightarrow\{0,1\}$} and a
coalition $J$ (the junta),  the junta function
\quad\mbox{$\JUNTA{f}:\{0,1\}^n\rightarrow\{0,1\}$}. It is derived
from $f$ in the following way:
    \[\JUNTA{f}(x)=\maj\left\{f(y)~|~y_{_J}=x_{_J}\right\}.\]
I.e., for a given input, $\JUNTA{f}$ reads only the votes of the
junta members, iterates over all the possible
    votes for the members outside the junta, and returns the more frequent
    result (assuming uniform distribution over the votes of the voters outside $J$).\\
} We define $\JUNTA{f}$ and $\JUNTA{g}$ with regard to the junta of
all the voters with small ignorability for either $f$ or $g$.
    \[ J=       \left\{\begin{array}{c|c}i & \OSI{i}{f}\leqslant\frac{\Delta}{n}\end{array}\right\}
        \cup    \left\{\begin{array}{c|c}i & \OSI{i}{g}\leqslant\frac{\Delta}{n}\end{array}\right\},\]
We prove that $\JUNTA{f}$ and $\JUNTA{g}$ are close to $f$ and $g$,
respectively and that there exists
    an issue aggregation function $h^\star$ such that \TRIPLET{\JUNTA{f}}{\JUNTA{g}}{h^\star} is a
    consistent aggregation mechanism that is close to $F$.
\end{proof-sketch}%
~\\\indent
There is a known characterization of the consistent independent
aggregation mechanism for the conjunction agenda. (This
characterization is a direct corollary from a series of works in the
more general framework of aggregation, E.g., \cite{Nehring2008},
\cite{Dokow2009}. We include a proof of it in the appendix)
\begin{lemma}\label{Lemma_mAND_WzeroIsUnanimity}~\\
Let \mbox{$f,g,h:\{0,1\}^n\rightarrow\{0,1\}$} be three voting
functions satisfying \mbox{$IC(\TRIPLET{f}{g}{h})=0$}. Then either
\mbox{$f=h\equiv0$}, or \mbox{$g=h\equiv0$}, or
\mbox{$f=g=h\in\OLIG$}.
\end{lemma}
A corollary from this theorem and theorem \ref{Main_Theorem_mAND} is
a characterization of the approximate aggregation mechanisms for
this agenda. Actually ,in the proof of theorem
\ref{Main_Theorem_mAND} we get a tighter characterization that
distinguishes between the two cases of consistent independent
aggregation mechanism.

\subsection{Xor Agenda}
For the agenda \mXORagenda we prove:
\begin{theorem}\label{Main_Theorem_mXOR}
Let $m\geqslant 3$ and let the agenda be
    \mbox{$\Agenda=\left<A^1,\ldots,A^{m-1},\xor\limits_{j=1}^{m-1}A^j\right>$}.

For any \mbox{$\epsilon\!<\!\frac{1}{6}$} and any independent
aggregation mechanism $F$:
    If \mbox{$IC(F)\!\leqslant\!\epsilon$}, then there exists an
    aggregation mechanism $G$ that satisfies consistency and
    independence such that \mbox{$d(F,G)\!\leqslant\!m\epsilon$}.
\end{theorem}
{\bf Technique}\footnote{The proof is similar to the analysis of the
BLR (Blum-Luby-Rubinfeld) linearity test done in
\cite{Bellare1995}.}: The proof uses the Fourier representation of
the issue aggregating functions. That is, representing the functions
as linear combinations of the linear boolean functions.
\begin{proof-sketch}~\\
\indent Given an independent aggregation mechanism
\mbox{$F=\TRIPLET{f}{g}{h}$} we analyze the expression
    \mbox{$\Exp[f(x)g(y)h(xy)]$} when $x$ and $y$ are sampled uniformly and independently.
On one hand we show that \mbox{$\Exp[f(x)g(y)h(xy)]=1\!-\!2IC(F)$}.
On the other hand we show that
\mbox{$\Exp[f(x)g(y)h(xy)]=\sum\limits_{\chi\in\LIN}\fhat(\chi)\ghat(\chi)\hhat(\chi)$}
when  $\left|\fhat(\chi)\right|$ equals
\mbox{$1\!-\!2\min\left(d\left(f,\chi\right),d\left(f,-\chi\right)\right)$}.
Hence, when $IC(F)$ is small then this sum is close to one and hence
there exists a linear function such that $f$,$g$, and $h$ are close
to it (up to negation). Noticing that for any linear function
$\chi$, \TRIPLET{\chi}{\chi}{\chi} and the permutations of
\TRIPLET{-\chi}{-\chi}{\chi} are consistent independent aggregation
mechanism for this agenda gives us the result.
\end{proof-sketch}

\subsection{Extending to $\delta$-independence Results}
\begin{theorem}\label{Thm_RelaxingIIA}~\\
If
\begin{IndentSection}
\noindent there exists a function \mbox{$\delta(\epsilon,n)$} s.t.
for any
    \mbox{$\epsilon>0$} and \mbox{$n\geqslant 1$}, if $F$ is an
    aggregation mechanism for \Agenda over $n$ voters satisfying
    independence and \mbox{$IC(F)\!\leqslant\!\delta(\epsilon)$}, then
    there exists an aggregation mechanism $G$ that satisfies
    consistency and independence such that \mbox{$d(F,G)\!<\!\epsilon$}.
\end{IndentSection}
\noindent Then,
\begin{IndentSection}
    \noindent for any \mbox{$\epsilon>0$} and \mbox{$n\geqslant 1$}, there
    exist \mbox{$\delta_{_{IC}},\delta_{_{DI}}>0$}, such that if $F$ is an
    aggregation mechanism for \Agenda over $n$ voters satisfying
    \mbox{$IC(F)\leqslant\delta_{_{IC}}$} and
    \mbox{$DI(F)\leqslant\delta_{_{DI}}$}, then there exists an
    aggregation mechanism $G$ that satisfies consistency and
    independence such that \mbox{$d(F,G)<\epsilon$}.\\
    Moreover, one can take
    \mbox{$\delta_{_{IC}}=\delta\left(\left(1-\beta_\epsilon\right)\epsilon\right)-\beta_\epsilon\epsilon$}
    and \mbox{$\delta_{_{DI}}=\frac{1}{2m}\beta_\epsilon\epsilon$}
    for any \mbox{$\beta_\epsilon\in[0,1]$} satisfying
    \mbox{$\delta\left(\left(1-\beta_\epsilon\right)\epsilon\right)\geqslant\beta_\epsilon\epsilon$}
\end{IndentSection}
\end{theorem}
In order to extend the results for the $\delta$-dependent case
\mbox{($DI(F)\neq0$)} we prove the following agenda-independent
proposition.
\begin{PropositionWithReference}{Proposition_DIisEpsIndependent}
Let $F$ be an aggregation mechanism for an agenda over $m$ issues
that satisfies \mbox{$DI(F)\leqslant\delta$}.
Then there exists an independent aggregation mechanism $H$ that
satisfies \mbox{$d(F,H)\leqslant 2m\delta$}.
\end{PropositionWithReference}
I.e., if $F$ is $\delta$-independent we can find a close consistent
aggregation mechanism $H$ and since it is close we can deduce bounds
on the proximity of $F$ to the consistent and independent
aggregation mechanisms from bounds on this proximity of $H$.
Similarly, since $H$ is close to $F$, we can deduce that if $F$ is
$\delta$-consistent than $H$ is $\delta'$-consistent for $\delta'$
close to $\delta$. Combining these we get the theorem.
\section{Summary and Future Work}
In this paper we defined the issue of approximate aggregation which
is a generalization of the study of aggregation mechanisms that
satisfy consistency and independence. We defined measures for the
relaxation of the consistency constraint (inconsistency index $IC$)
and for the relaxation of the independence constraint (dependency
index $DI$). To our knowledge, this is the first time this question is stated in its general form.

We proved that relaxing these constraints does not extend the set of
satisfying aggregation mechanisms in a non-trivial way for any
truth-functional agenda in which every conclusion is either conjunction or xor up to negation of inputs or output. We saw that every conclusion of two premises can be stated as such. We also saw that any affine agenda can be represented as truth-functional agenda with xor conclusions only and derived a better approximate aggregation characterization for this family.
Particulary we calculated the dependency between the extension of this class ($\epsilon$) and the inconsistency index ($\delta(\epsilon)$) (although probably not strictly) for two families of truth-functional agendas with one conclusion. The relation we proved includes dependency on the number of voters ($n$). In the works that preceded us for preference agendas
(\cite{Kalai2002}, \cite{Mossel2010}, \cite{Keller2010}) the
relation did not include such a dependency. An interesting question
is whether such a dependency is inherent for conjunction agendas
or whether it is possible to prove a relation that does not
depend on $n$.

A major assumption in this paper is the uniform distribution over
the inputs which is equivalent to assuming i.i.d uniform
distribution over the premises. We think that our results can be
extended for other distributions (still assuming voters' opinions
are distributed i.i.d) over the space over premises' opinions which
seem more realistic.

Immediate extensions for this work can be to extend our result to
more complex truth-functional agendas and generalize our results to non-truth-functional agendas to get a result unifying our work and Kalai, Mossel, and Keller's works for the preference agenda.

A major open question is whether one can find an agenda for which
relaxing the constraints of independence and consistency extends the
class of satisfying aggregation mechanisms in a non-trivial way.
\bibliographystyle{plain}

\clearpage
\onecolumn
\appendix
\section{Lemmas Proof - General} \label{Appendix_LemmasProof_General}
\subsection{Propositions \ref{Prop_FuncDistToMechDist},\ref{Prop_GenWisLifshitz}}
For a given pair of independent aggregation mechanisms, the
following propositions connect between the pairwise distance between
respective issue-aggregating functions (which we found easier to
analyze in most cases) and both the distance between the mechanisms
and the inconsistency indices of them.
\begin{proposition}\label{Prop_FuncDistToMechDist}~\\
For any agenda \Agenda of $m$ issues and any voting
functions\\\indent
    \mbox{$f^1,\ldots,f^m,g^1,\ldots,g^m,:\{0,1\}^n\rightarrow\{0,1\}$},
    \[d^\Agenda(\left<f^1,\ldots,f^m\right>~,~\left<g^1,\ldots,g^m\right>)\leqslant
        \sum\limits_{j=1}^mPr\left[f^j(X^j)\neq g^j(X^j)~|~X\in\Agenda^n\right].\]
\end{proposition}
\begin{proof-of-proposition}{}
Direct use of the union-bound inequality.
\end{proof-of-proposition}%
\begin{proposition}~\\
For any agenda \Agenda of $m$ issues and voting functions
    \mbox{$f^1,\ldots,f^m,g^1:\{0,1\}^n\rightarrow\{0,1\}$},
    \[\left|IC^\Agenda(\left<f^1,f^2,\ldots,f^m\right>)-IC^\Agenda(\left<g^1,f^2,\ldots,f^m\right>)\right|\leqslant
        Pr[f^1(X^1)\neq g^1(X^1)~|~X\in\Agenda^n].\]
\end{proposition}%
\begin{proof-of-proposition}{}~\\
$\begin{array}[t]{rlr}%
IC(\left<f^1,\ldots,f^m\right>)
    &= \Pr\left[\left(f^1(X^1),f^2(X^2),\ldots,f^m(X^m)\right)\notin\Agenda~|~X\in\Agenda^n\right]\\
    &\leqslant Pr[f^1(X^1)\neq g^1(X^1)~|~X\in\Agenda^n]\\&\qquad
                + \Pr\left[\left(f^1(X^1),\ldots,f^m(X^m)\right)\notin\Agenda
                    ~\bigwedge f^1(x)= g^1(x)~|~X\in\Agenda^n\right]\\
    &\leqslant Pr[f^1(X^1)\neq g^1(X^1)~|~X\in\Agenda^n]\\&\qquad
                + \Pr\left[\left(g^1(X^1),f^2(X^2),\ldots,f^m(X^m)\right)\notin\Agenda~|~X\in\Agenda^n\right]\\
    &= IC^\Agenda(\left<g^1,\ldots,f^m\right>) +  Pr[f^1(X^1)\neq g^1(X^1)~|~X\in\Agenda^n]&
\end{array}$\\
Hence, \mbox{$IC(f,g,h)-IC(f',g,h)\leqslant  Pr[f^1(X^1)\neq g^1(X^1)~|~X\in\Agenda^n]$}.\\Similarly we
can prove that \mbox{$IC(f',g,h)-IC(f,g,h)\leqslant  Pr[f^1(X^1)\neq g^1(X^1)~|~X\in\Agenda^n]$}.
\end{proof-of-proposition}%
As a corollary of the above we get
\begin{proposition}~\label{Prop_GenWisLifshitz}\\
For any agenda \Agenda of $m$ issues and any voting
functions\\\indent
    \mbox{$f^1,\ldots,f^m,g^1,\ldots,g^m,:\{0,1\}^n\rightarrow\{0,1\}$},
    \[\left|IC^\Agenda(\left<f^1,\ldots,f^m\right>)-IC^\Agenda(\left<g^1,\ldots,g^m\right>)\right|\leqslant
        \sum\limits_{j=1}^mPr\left[f^j(X^j)\neq g^j(X^j)~|~X\in\Agenda^n\right].\]
\end{proposition}

\subsection{Proposition \ref{Proposition_DIJisEpsIndependent}}
\begin{propositionUN}~\\
Let $F$ be an aggregation mechanism and $j$ an issue.
If \mbox{$DI^j(F)\leqslant\epsilon$}, then there exists an
aggregation mechanism $H$ that satisfies $DI^j(H)=0$ and
\mbox{$d(F,H)\leqslant 2\epsilon$}.
If \mbox{$DI^j(F)\geqslant\epsilon$}, then every aggregation
mechanism $H$ that satisfies $DI^j(H)=0$, also satisfies
\mbox{$d(F,H)\geqslant \half\epsilon$}\\
\end{propositionUN}
\noindent
\begin{proof-of-proposition}{}~\\
With no loss of generality assume that $j=1$.
\begin{itemize}
\item
Let $F$ be an aggregation mechanism. We define the functions
\mbox{$G^1,\ldots,G^m:\FeasableAssignments^n\rightarrow\{0,1\}$} by:
    \[\begin{array}{ll}
        &G^1(X)=\left\{\begin{array}{rl}
        1 & \Pr\limits_{Y\in\FeasableAssignments^n}\left[\left(F(X)\right)^1=1~|~Y^1=X^1\right]\geqslant\half\\
        0 & \mbox{otherwise}
    \end{array}\right.\\
    j=2,\ldots,m & G^j(X)=\left(F(X)\right)^j
    \end{array}\]
and an aggregation mechanism
\mbox{$G(X)=\left<G^1(X),\ldots,G^m(X)\right>$}.
Clearly \mbox{$DI^1(G)=0$}.\\
$\begin{array}{rll}%
d(F,G)
    &= \Pr\limits_{X\in\FeasableAssignments^n}\left[(F(X))^1\neq G^1(X)\right]\\
    &= \Pr\limits_{X\in\FeasableAssignments^n}\left[
        \Pr\limits_{Y\in\FeasableAssignments^n}\left[(F(X))^1\neq (F(Y))^1 | X^1=Y^1\right]\geqslant \half
        \right]\\
    &\leqslant 2\Exp\limits_{X\in\FeasableAssignments^n}\left[
        \Pr\limits_{Y\in\FeasableAssignments^n}\left[(F(X))^1\neq (F(Y))^1 | X^1=Y^1\right]
        \right]\\
    &= 2DI^1(F)
        &\qedhere
\end{array}$
\item
Let $F$ be an aggregation mechanism that is $\epsilon$-close to
satisfy $DI^1(F)=0$. That is, we can find an aggregation mechanism
$G$ such that $d(F,G)\leqslant\epsilon$ and $DI^1(G)=0$.\\
$\begin{array}{rll}%
DI^1(F)
    &= \Exp\limits_{X\in\FeasableAssignments^n}\left[
               \Pr\limits_{Y\in\FeasableAssignments^n}\left[(F(X))^1\neq (F(Y))^1 | X^1=Y^1\right]
               \right]\\
    &\leqslant \Pr[F(X)\neq G(X)]
        + \sum\limits_{X:F(X)= G(X)}\Pr\limits_{Y\in\FeasableAssignments^n}\left[(G(X))^1\neq (F(Y))^1 |
        X^1=Y^1\right]\\
    &\leqslant \epsilon
        + \sum\limits_{X:F(X)= G(X)}\Pr\limits_{Z\in\FeasableAssignments^n}[Z=X]
                \Pr\limits_{Y\in\FeasableAssignments^n}\left[(G(Y))^1\neq (F(Y))^1 | X^1=Y^1\right]\\
    &\leqslant 2\epsilon
        &\qedhere
\end{array}$
\end{itemize}
\end{proof-of-proposition}
\subsection{Proposition \ref{Proposition_DIisEpsIndependent}}
\begin{propositionUN}~\\
Let $F$ be an aggregation mechanism for an agenda over $m$ issues
that satisfies \mbox{$DI(F)\leqslant\delta$}.
Then there exists an independent aggregation mechanism $H$ that
satisfies \mbox{$d(F,H)\leqslant 2m\delta$}.
\end{propositionUN}
\noindent
\begin{proof-of-proposition}{}~\\
We define issue aggregating functions
\mbox{$f^1,\ldots,f^m:\{0,1\}^n\rightarrow\{0,1\}$} by:
    \[f^j(t)=\left\{\begin{array}{rl}
        1 & \Pr\limits_{X\in\FeasableAssignments^n}\left[G^j(X)=1~|~X^j=t\right]\geqslant\half\\
        0 & \mbox{otherwise}
    \end{array}\right.\]
and an (independent) aggregation mechanism
\mbox{$F=\left<f^1,\ldots,f^m\right>$}.\\
$\begin{array}{rll}%
d(F,G)
    &= \Pr\limits_{X\in\FeasableAssignments^n}\left[F(X)\neq G(X)\right]\\
    &\leqslant \sum\limits_{j=1}^m \Pr\limits_{X\in\FeasableAssignments^n}\left[F^j(X)\neq G^j(X)\right]\\
    &= \sum\limits_{j=1}^m \Pr\limits_{X\in\FeasableAssignments^n}\left[
        \Pr\limits_{Y\in\FeasableAssignments^n}\left[G(X)\neq G(Y) | X^j=Y^j\right]\geqslant \half
        \right]\\
    &\leqslant \sum\limits_{j=1}^m 2\Exp\limits_{X\in\FeasableAssignments^n}\left[
        \Pr\limits_{Y\in\FeasableAssignments^n}\left[G(X)\neq G(Y) | X^j=Y^j\right]
        \right]\\
    &\leqslant  2m\delta_{_{DI}}
        &\qedhere\end{array}$
\end{proof-of-proposition}

\subsection{Id Agenda}
For completeness we add here an approximate aggregation theorem for the id agenda \PAIR{A}{A}
\begin{theorem}~\\%
For any \mbox{$\epsilon>0$} and any independent aggregation mechanism $F$:\\
    If \mbox{$IC^\PAIR{A}{A}\leqslant\epsilon$}, then there exists an
    aggregation mechanism $G$ that satisfies consistency and
    independence such that \mbox{$d(F,G)\leqslant\epsilon$}.
\end{theorem}
\begin{proof}
This theorem is trivial since
    \[IC^\PAIR{A}{A}(\PAIR{f}{g})
        ~= \Pr\left[f(x)\neq g(y)~|~x=y\right]
        ~= \Pr\left[f(x)\neq g(x)\right]
        ~= d(f,g)\]
Noticing that any aggregation mechanism of the form \PAIR{f}{f} is
consistent for this agenda so we get the theorem.
\end{proof}
\section{Lemmas Proof - Conjunction agenda} \label{Appendix_LemmasProof_mAND}
\label{Chapter Lemmas proof_mAND}
\subsection{Theorem \ref{Main_Theorem_mAND}}%
\begin{theoremUN}~\\
Let $m\geqslant 2$ and let the agenda be
    $\Agenda=\mANDagenda$.

For any \mbox{$\epsilon>0$} and any independent aggregation mechanism $F$:\\
    If \mbox{$IC(F)\leqslant\epsilon$}, then there exists an
    aggregation mechanism $G$ that satisfies consistency and
    independence such that
\mbox{$d(F,G)<5m \left(n^2\epsilon\right)^\frac{1}{m^2+m-1}$}.
\end{theoremUN}%
\begin{proof}~\\
For proving this  bound, we define for a given function
\mbox{$f:\{0,1\}^n\rightarrow\{0,1\}$} and a coalition $J$ (the
junta),  the junta function
\quad\mbox{$\JUNTA{f}:\{0,1\}^n\rightarrow\{0,1\}$}. It is derived
from $f$ in the following way:
    \[\JUNTA{f}(x)=\maj\left\{f(y)~|~y_{_J}=x_{_J}\right\}.\]
I.e., for a given input, $\JUNTA{f}$ reads only the votes of the
junta members, iterates over all the possible
    votes for the members outside the junta, and returns the more frequent
    result (assuming uniform distribution over the votes of the voters outside $J$).\\
We prove the following lemma:
\begin{lemma}~\\\label{Lemma_mAND_WisSmallSoFeqGinUnanLog}%
Let
\mbox{$f^1,\ldots,f^m:\left\{0,1\right\}^n\rightarrow\left\{0,1\right\}$}
be $m$ voting functions. Define \[\ICtilde(f^1,\ldots,f^m) =
\min\limits_hIC^\mANDagenda(f^1,\ldots,f^m,h).\] If there exists
constants $\Delta$ and $\epsilon$ such that:
    \[\begin{array}{@{\star~~}l}
        \ICtilde(f^1,\ldots,f^m)\leqslant\epsilon\\
        \forall j~:~d(f^j,0)\geqslant\Delta\\
        \epsilon<2^{-m^2-3}m^{-1}n^{-2}\Delta^{m^2+m-1}%
    \end{array}\]
Then there exist a coalition $J\subseteq\{1,\ldots,n\}$
    such that the junta functions \JUNTA{f^j}
satisfy
    \[\begin{array}{@{\star~~}l}
        \mbox{There exists an oligarchy $g\in\OLIG$ s.t. $\forall j~\JUNTA{f^j}=g$}\\
        \forall j~:d(f^j,\JUNTA{f^j})\leqslant 4n^2\epsilon\Delta^{1-m}\\
        |J|\leqslant m\left(1+\log_2\frac{1}{\Delta}\right)\\
    \end{array}\]
\end{lemma}
The requested bound is a corollary of lemma
\ref{Lemma_mAND_WisSmallSoFeqGinUnanLog}.\\
Assume a mechanism \mbox{$F=\left<f^1,\ldots,f^m,h\right>$} is given
such that {$IC(F)\leqslant \epsilon$}.
Then \mbox{$\ICtilde(f^1,\ldots,f^m)\leqslant \epsilon$}. Define
\mbox{$\Delta= 4\left(n^2\epsilon\right)^\frac{1}{m^2+m-1}$}. If
there exists $j\in\{1,\dots,m\}$ (with no loss of generality assume
$j=1$) s.t. \mbox{$d(f^j,0)<\Delta$}, then
\mbox{$\left<f^1,f^2,\ldots,f^m,h\right>$} is $\Delta$-close to
\mbox{$F=\left<0,f^2,\ldots,f^m,h\right>$} and
\mbox{$\left(\epsilon+\Delta\right)$-close}
     to \mbox{$F=\left<0,f^2,\ldots,f^m,0\right>$} which is a consistent mechanism.
If \mbox{$\forall j\in\{1,\ldots,m\}~ d(f^j,0) \geqslant\Delta$},
then \mbox{$\left<f^1,f^2,\ldots,f^m,h\right>$} is
\mbox{$4mn^2\epsilon\Delta^{1-m}$-close} to
\mbox{$\left<g,\ldots,g,h\right>$} for some oligarchy $g$ and
\mbox{$\left(\epsilon+4mn^2\epsilon\Delta^{1-m}\right)$-close} to
\mbox{$\left<g,\ldots,g,g\right>$} which is a consistent mechanism.
Since \mbox{$\max(
    \epsilon+\Delta,
    \epsilon+4mn^2\epsilon\Delta^{1-m}
)\leqslant 5m \left(n^2\epsilon\right)^\frac{1}{m^2+m-1}$} (when
$n^2\epsilon<1$), we get the theorem.
\end{proof}%
\subsection{Lemma \ref{Lemma_mAND_WisSmallSoFeqGinUnanLog}}%
\begin{lemmaUN}~\\
Let
\mbox{$f^1,\ldots,f^m:\left\{0,1\right\}^n\rightarrow\left\{0,1\right\}$}
be $m$ voting functions. Define \[\ICtilde(f^1,\ldots,f^m) =
\min\limits_hIC^\mANDagenda(f^1,\ldots,f^m,h).\] If there exists
constants $\Delta$ and $\epsilon$ such that:
    \[\begin{array}{@{\star~~}l}
        \ICtilde(f^1,\ldots,f^m)\leqslant\epsilon\\
        \forall j~:~d(f^j,0)\geqslant\Delta\\
        \epsilon<2^{-m^2-3}m^{-1}n^{-2}\Delta^{m^2+m-1}%
    \end{array}\]
Then there exist a coalition $J\subseteq\{1,\ldots,n\}$
    such that the junta functions \JUNTA{f^j}
satisfy
    \[\begin{array}{@{\star~~}l}
        \mbox{There exists an oligarchy $g\in\OLIG$ s.t. $\forall j~\JUNTA{f^j}=g$}\\
        \forall j~:d(f^j,\JUNTA{f^j})\leqslant 4n^2\epsilon\Delta^{1-m}\\
        |J|\leqslant m\left(1+\log_2\frac{1}{\Delta}\right)\\
    \end{array}\]
\end{lemmaUN}%
\begin{proof-of-lemma}{}~\\
The proof of the lemma is constructive and defines the junta $J$. We
define the junta to be all the voters with small ignorability for at
least one of the functions
    \[ J=\bigcup\limits_{j=1}^m\left\{\begin{array}{c|c}
        i &
        \OSI{i}{f^j}\leqslant\frac{\Delta}{n}\end{array}\right\}\]
 and prove for this junta the different claims of the lemma.
\begin{itemize}
\item $\forall j~:d(f^j,\JUNTA{f^j})\leqslant
4n^2\epsilon\Delta^{1-m}$\\
The following lemma states a connection between the influence of the
voter on $f$($\Inf{i}{f}$), the ignorability of the same voter for
$g$($\OSI{i}{g}$), and the inconsistency index of $f$ and $g$.
\begin{lemma}\label{Lemma_mAND_BoundPI}
Let \mbox{$f^1,\ldots,f^m:\{0,1\}^n\rightarrow\{0,1\}$} be $m$
voting functions , \mbox{$i\in\{1,\ldots,n\}$} be a voter, and
\mbox{$k,l\in\{1,\ldots,m\}$} two different issues. Then
\mbox{$\OSI{i}{f^k}\cdot\Inf{i}{f^l}\leqslant
        4   \left(\prod\limits_{j\neq k,l}d(f^j,0)\right)^{-1}
        \ICtilde(f^1,\ldots,f^m)$}
\end{lemma}%
 Using this lemma, we can bound the influence of the voters
    outside the Junta by \mbox{$\Inf{i}{f}\leqslant 4n\epsilon \Delta^{1-m}$}.\\
    Since all these voters have small influence, the function $\JUNTA{f}$ cannot be too far
    from the original function $f$.
    \begin{lemma}\label{Lemma_JuntaIsClose}
Let \mbox{$f:\{0,1\}^n\rightarrow\{0,1\}$} be a binary function and
\mbox{$J\subseteq\{1,\ldots,n\}$} a coalition. Then
    \mbox{$\DISTx{f}{\JUNTA{f}}\leqslant~\sum\limits_{i\notin J}\Inf{i}{f}$}.
    \end{lemma}%
    Combining both gives us the desired bound.
\item $|J|\leqslant m\left(1+\log_2\frac{1}{\Delta}\right)$\\
    Since $f$ is $\Delta$-far from zero we can bound the number of voters that have small
    ignorability.
\begin{lemma}\label{Lemma_JuntaIsSmall}
Let \mbox{$f:\{0,1\}^n\rightarrow\{0,1\}$} be a voting function. If
\mbox{$d(f,0)\geqslant\Delta$} then at most
\mbox{$\left(1+\log_2\frac{1}{\Delta}\right)$} voters have the
property:~
    \mbox{$\OSI{i}{f}\leqslant\frac{\Delta}{n}$}.%
\end{lemma}%
\item $\JUNTA{f}=\JUNTA{g}\in\OLIG$\\
    Both $\JUNTA{f}$ and $\JUNTA{g}$ depend on a small number of voters so the inconsistency index has large
    granularity
\begin{lemma}\label{Lemma_mAND_WhasGranularity}
Let \mbox{$f^1,\ldots,f^m:\{0,1\}^n\rightarrow\{0,1\}$} be $m$
voting functions that depend only on the votes of the members of
$J$. Then \mbox{there exists an integer $C$ s.t.
$\ICtilde(f,g)=C\cdot 2^{-m|J|}$}.
\end{lemma}%
    So if we show that
    \[ \epsilon+\sum\limits_{j=1}^md(f^j,\JUNTA{f^j})<2^{-m|J|}\]
    then we get that $\ICtilde(\JUNTA{f},\JUNTA{g})=0$ and based on lemma \ref{Lemma_mAND_WzeroIsUnanimity}
    we get that there exists an oligarchy  of a sub-coalition of $J$
    $g$ s.t. $\forall j~\JUNTA{f^j}=g$.
    Notice that $\JUNTA{f^j}$ cannot be the constant zero function for any $j$ since
    \mbox{$d(f^j,0)\geqslant \Delta$} and
    \mbox{$d(f,\JUNTA{f^j})\leqslant 4n^2\Delta^{1-m} \leqslant 2^{-m^2-1}m^{-1}\Delta^{m^2}<\Delta
    $}.
\[\begin{array}{l}
\epsilon+\sum\limits_{j=1}^md(f^j,\JUNTA{f^j})
\leqslant \epsilon\left(1+4mn^2\Delta^{1-m}\right)
\leqslant 8mn^2\epsilon \Delta^{1-m}
< 2^{-m^2}\Delta^{m^2}
\leqslant 2^{-m|J|}
\end{array}.\]

\end{itemize}
\end{proof-of-lemma}%
\subsection{Lemma \ref{Lemma_mAND_BoundPI}}%
\begin{lemmaUN}~\\
Let \mbox{$f^1,\ldots,f^m:\{0,1\}^n\rightarrow\{0,1\}$} be $m$
voting functions , \mbox{$i\in\{1,\ldots,n\}$} be a voter, and
\mbox{$k,l\in\{1,\ldots,m\}$} two different issues. Then
\mbox{$\OSI{i}{f^k}\cdot\Inf{i}{f^l}\leqslant
        4   \left(\prod\limits_{j\neq k,l}d(f^j,0)\right)^{-1}
        \ICtilde(f^1,\ldots,f^m)$}
\end{lemmaUN}%
\begin{proof-of-lemma}{}~\\
Let \mbox{$h:\{0,1\}^n\rightarrow\{0,1\}$} be a voting function s.t.
\mbox{$IC(f^1,\ldots,f^M,h)=\ICtilde(f^1,\ldots,f^m)$}.\\
(We use the notation $x \xor e_i$ for adding $e_i$ (the \ith elementary
vector) which is equivalent to flipping the \ith bit
$0\leftrightarrow1$)\\
$\begin{array}{rl}%
\ICtilde(f^1,\ldots,f^m)
    &= IC(f^1,\ldots,f^M,h)\\
    &= \Pr\left[\bigwedge\limits_{j=1}^mf^j(x^j)
        \neq h\left(\bigwedge\limits_{j=1}^mx^j\right)\right]\\
    &= \Pr[\left(x^k\right)_i=0]\prod\limits_{j\neq k,l}
            \Pr[f^j(x^j)=1]\cdot
        \Pr\left[\!\!\!\begin{array}{c|l}
            \bigwedge\limits_{j=1}^mf^j(x^j)\neq h\left(\bigwedge\limits_{j=1}^mx^j\right)
            &
            \begin{array}{l}\left(x^k\right)_i=0 \\ \forall j\neq k,l~:f^j(x^j)=1\end{array}
        \end{array}\!\!\!\!\!\right]\\
    &= \half\prod\limits_{j\neq k,l} d(f^j,0)\cdot
        \Pr\left[\begin{array}{c|l}
            f^k(x^k)\wedge f^l(x^l)\neq h\left(\bigwedge\limits_{j=1}^mx^j\right)
            &
            \begin{array}{l}\left(x^k\right)_i=0 \\ \forall j\neq k,l:f^j(x^j)=1\end{array}
        \end{array}\right]\\
    &\geqslant \half\prod\limits_{j\neq k,l} d(f^j,0)\cdot \half\cdot
        \Pr\left[\!\!\!\!\!\begin{array}{c@{\!}|@{\!}l}
            \begin{array}{c}
            f^k(x^k)\wedge f^l(x^l)\neq h\left(\bigwedge\limits_{j\neq l}x^j\wedge x^l\right)\\
                \bigvee\\
            f^k(x^k)\wedge f^l(x^l\xor e_i)\neq h\left(\bigwedge\limits_{j\neq l}x^j\wedge (x^l\xor e_i)\right)\\
            \end{array}
            &
            \begin{array}{l}\left(x^k\right)_i=0 \\ \forall j\neq k,l:f^j(x^j)=1
            \end{array}
        \end{array}\!\!\!\!\!\right]\\
    &\geqslant\frac{1}{4}
    \prod\limits_{j\neq k,l} d(f^j,0)\cdot
        \Pr\left[\begin{array}{c|l}
            f^k(x^k)\wedge f^l(x^l)\neq f^k(x^k)\wedge f^l(x^l\xor e_i)
            &
            \begin{array}{l}\left(x^k\right)_i=0 \\ \forall j\neq k,l:f^j(x^j)=1\end{array}
        \end{array}\right]\\
    &\geqslant\frac{1}{4}
    \prod\limits_{j\neq k,l} d(f^j,0)\cdot
        \Pr\left[\begin{array}{c|l}
            \begin{array}{c}f^k(x^k)=1 \\ f^l(x^l)\neq f^l(x^l\xor e_i) \end{array}
            &
            \begin{array}{l}\left(x^k\right)_i=0 \\ \forall j\neq k,l:f^j(x^j)=1\end{array}
        \end{array}\right]\\
    &=\frac{1}{4}
    \prod\limits_{j\neq k,l} d(f^j,0)\cdot
        \OSI{i}{f^k}\Inf{i}{f^l}
\end{array}$
\end{proof-of-lemma}%
\subsection{Lemma \ref{Lemma_JuntaIsClose}}%
\begin{lemmaUN}~\\
Let \mbox{$f:\{0,1\}^n\rightarrow\{0,1\}$} be a binary function and
\mbox{$J\subseteq\{1,\ldots,n\}$} a coalition. Then
    \mbox{$\DISTx{f}{\JUNTA{f}}\leqslant~\sum\limits_{i\notin J}\Inf{i}{f}$}.
\end{lemmaUN}%
\begin{proof-of-lemma}{}~\\
We define for a vector $c\in\{0,1\}^J$ the function
    \mbox{$f^J_c:\{0,1\}^n\rightarrow\{0,1\}$}
by $f^J_c(x)=f(y)$ where $y_{_J}=c$ and $y_{_{-J}}=x_{_{-J}}$.
Assume that $c_i$ is sampled according to $p$. Then
\mbox{$\JUNTA{f}(x_{_J},x_{_{-J}}) = \left\{\begin{array}{ll}
                            0 & \Exp_c[f^J_c(x)]<\half\\
                            1 & \Exp_c[f^J_c(x)]\geqslant\half
                        \end{array}\right.$}
~\\
We will use the following isoperimetric inequality on the boolean
cube:
\begin{propositionUN}[The Isoperimetric Inequality for The Boolean
Cube \cite{Bollob'as1986}]~\\
Let \mbox{$f:\{0,1\}^n\rightarrow\{0,1\}$} be a boolean function.
Then \mbox{$\sum\limits_i\Inf{i}{f}~\geqslant~\min(\Exp[f],1-\Exp[f])$}.%
\end{propositionUN}
$\begin{array}{rlr}%
\mbox{For any $c\in\{0,1\}^J$~:~} &
    \begin{array}[t]{rl}
    \sum\limits_{i\notin J}\Inf{i}{f^J_c}
    &= \sum\limits_i\Inf{i}{f^J_c}\\
    &\geqslant \min(\Exp[f^J_c],1-\Exp[f^J_c])
    \end{array}\\
\mbox{For $i\notin J$:} &
    \begin{array}[t]{rl}
    \Inf{i}{f}
        & = \Pr[f(x)\neq f(x \xor e_i)]\\
        & = \Exp_c\left[\Pr[f^J_c(x)\neq f^J_c(x \xor e_i)]\right]\\
        & = \Exp_c\left[\Inf{i}{f^J_c}\right]
    \end{array}\\
~\\
\Exp_c\left[\sum\limits_{i\notin J}\Inf{i}{f^J_c}\right]
    &= \sum\limits_{i\notin J}\Exp_c\left[\Inf{i}{f^J_c}\right]\\
    &= \sum\limits_{i\notin J}\Inf{i}{f}\\
\Exp_c\left[\sum\limits_{i\notin J}\Inf{i}{f^J_c}\right]
    &\geqslant \Exp_c\left[\min\left(\Exp\left[f^J_c\right],1-\Exp\left[f^J_c\right]\right)\right]\\
    &= Pr[\JUNTA{f}(x)\neq f(x)]\\
    &= d(\JUNTA{f},f)\\
~\\
\DISTx{f}{\JUNTA{f}}
    & \leqslant \sum\limits_{i\notin J}\Inf{i}{f}
&\qedhere\end{array}$
\end{proof-of-lemma}%
\subsection{Lemma \ref{Lemma_JuntaIsSmall}}%
\begin{lemmaUN}~\\
Let \mbox{$f:\{0,1\}^n\rightarrow\{0,1\}$} be a voting function. If
\mbox{$d(f,0)\geqslant\Delta$} then at most
\mbox{$\left(1+\log_2\frac{1}{\Delta}\right)$} voters have the
property:~
    \mbox{$\OSI{i}{f}\leqslant\frac{\Delta}{n}$}.%
\end{lemmaUN}%
\begin{proof-of-lemma}{}~\\
Define~
    \mbox{$J=\{i~|~\OSI{i}{f}\leqslant \frac{\Delta}{n}\}$}. Then:~~
\[\begin{array}[t]{rlr}%
Pr[f(x)=1]
    &\leqslant  Pr[x_{_J}=\bar{1}] + \sum\limits_{i\in J} Pr[f(x)=1~|~x_i\neq1]\cdot Pr[x_i\neq1]\\
    &\leqslant  2^{-|J|} + \sum\limits_{i\in J} \frac{1}{2}\cdot \OSI{i}{f}\\
    &\leqslant  2^{-|J|}+\frac{|J|}{2}\frac{\Delta}{n}\\
    &\leqslant  2^{-|J|}+\frac{\Delta}{2}\\
\Delta
    &\leqslant  2^{-|J|}+\frac{\Delta}{2}\\
|J|
    &\leqslant 1+\log_2\frac{1}{\Delta}
&\qedhere\end{array}\]
\end{proof-of-lemma}%
\subsection{Lemma \ref{Lemma_mAND_WhasGranularity}}%
\begin{lemmaUN}~\\
Let \mbox{$f^1,\ldots,f^m:\{0,1\}^n\rightarrow\{0,1\}$} be $m$
voting functions that depend only on the votes of the members of
$J$. Then \mbox{there exists an integer $C$ s.t.
$\ICtilde(f,g)=C\cdot 2^{-m|J|}$}.
\end{lemmaUN}%
\begin{proof-of-lemma}{}~\\
Let \mbox{$h:\{0,1\}^n\rightarrow\{0,1\}$} be a voting function s.t.
\mbox{$IC(f^1,\ldots,f^M,h)=\ICtilde(f^1,\ldots,f^m)$}.
Then,\\
$\begin{array}{rlr}%
\ICtilde(f^1,\ldots,f^m)
    &= \Pr\left[\bigwedge\limits_{j=1}^mf^j(x^j)\neq h\left(\bigwedge\limits_{j=1}^mx^j\right)\right]\\
    &= \sum\limits_{c^1,\dots,c^m\in\{0,1\}^J}
            \prod\limits_{j=1}^m\Pr\left[\left(x^j\right)_{_J}=c^j\right]~\cdot
           \Pr\left[\bigwedge\limits_{j=1}^mf^j(x^j)\neq h\left(\bigwedge\limits_{j=1}^mx^j\right)
               ~|\forall j~\left(x^j\right)_{_J}=c^j\right]\\
    &= 2^{-m|J|}\#\left\{c^1,\dots,c^m\in\{0,1\}^n~|~
            \left(\bigwedge\limits_{j=1}^mf^j(x^j)\neq h\left(\bigwedge\limits_{j=1}^mx^j\right)\right)
                \bigwedge
            \left(\forall j~\left(x^j\right)_{_J}=\bar{0}\right)
        \right\}
\qedhere\end{array}$
\end{proof-of-lemma}%
\subsection{Lemma \ref{Lemma_mAND_WzeroIsUnanimity}}%
\begin{lemmaUN}~\\
Let \mbox{$f^1,\ldots,f^m,h:\{0,1\}^n\rightarrow\{0,1\}$} be
\mbox{$m+1$} voting functions satisfying
\mbox{$IC(\left<f^1,\ldots,f^m,h\right>)=0$}. Then either there
exists an issue $j$ s.t. \mbox{$f^j=h\equiv0$} or
\mbox{$f^1=f^2=\ldots=f^m=h\in\OLIG$}.
\end{lemmaUN}%
\begin{proof-of-lemma}{}~\\
Assume that for issues $j$, $f^j$ is not the constant zero function.
\mbox{We will prove that
$f^1\!=\!f^2\!=\!\ldots\!=\!f^m\!=\!h\in\OLIG$} by proving the
following series of claims.
\begin{itemize}
\item For all issues $j$ $f^j(\bar{1})=1$\\
With no loss of generality, assume for contraction that
\mbox{$f^1(\bar{1})=0$}. Let \mbox{$x\in\{0,1\}^n$}. Then
    \[h(x)  =h\left(\bar{1}\wedge\left(\bigwedge\limits_{j=2}^mx\right)\right)
            =f^1\left(\bar{1}\right)\wedge\left(\bigwedge\limits_{j=2}^m f^j\left(x\right)\right)
            =0.
    \]
I.e. \mbox{$h\equiv 0$}. From that we can conclude that there exists
an issue $j$ s.t. \mbox{$f^j\equiv0$} and get a contradiction.
\item For all issues $j$ $f^j=h$\\
We will prove that $f^1=h$. The proof is similar for all $j$.\\
Let \mbox{$x\in\{0,1\}^n$}. Then
    \mbox{$h(x)=h\left(x\wedge\left(\bigwedge\limits_{j=2}^m\bar{1}\right)\right)
            =f^1\left(x\right)\wedge\left(\bigwedge\limits_{j=2}^m f^j\left(\bar{1}\right)\right)
            =f^1(x)$}

\item $f^1\in\OLIG$\\
Let \mbox{$J=\{i~|~\Inf{i}{f^1}\neq0\}$}. Then $f^1$ is a function
of \mbox{$\{x_i\}_{i\in J}$}. Based on lemma
\ref{Lemma_mAND_BoundPI}, for $i\in J$ $\OSI{i}{f^1}=0$ and hence
    $\left[x_i=0~\Rightarrow f^1(x)=0\right]$.
So we get that $f^1$ is the oligarchy of $J$. \qedhere
\end{itemize}
\end{proof-of-lemma}%
\section{Lemmas Proof - XOR agenda} \label{Appendix_LemmasProof_mXOR}
\label{Chapter Lemmas proof_mXOR}
\subsection{Theorem \ref{Main_Theorem_mXOR}}%
\begin{theoremUN}~\\
Let $m\geqslant 3$ and let the agenda be
    \mbox{$\Agenda=\left<A^1,\ldots,A^{m-1},\xor\limits_{j=1}^{m-1}A^j\right>$}.

For any \mbox{$\epsilon\!<\!\frac{1}{6}$} and any independent
aggregation mechanism $F$:
    If \mbox{$IC(F)\!\leqslant\!\epsilon$}, then there exists an
    aggregation mechanism $G$ that satisfies consistency and
    independence such that \mbox{$d(F,G)\!\leqslant\!m\epsilon$}.
\end{theoremUN}%
\begin{proof}~\\
The theorem is a corollary of the following lemma:\\
(We rename the
values from $\{0,1\}$ to $\{1,-1\}$ in order to ease the analysis
(use multiplication instead of xor) and in particular use the
Fourier transformation for $f,g,h$.\footnote{ Fourier transforms are
widely used in mathematics, computer science, and engineering. The
main idea is representing the function $f$ over an orthonormal basis
$\chi_{_S}$ when the inner product is define to be
\mbox{$\left<f,g\right>=\Exp[f(x)g(x)]$} and the basis vectors
$\chi_{_S}$ are defined to be \mbox{$\chi_{_S}(x)=\prod\limits_{i\in
S}x_i$} for $S\subseteq\{1,\ldots,n\}$. The coefficients of $f$
according to the Fourier basis are notated $\fhat(S)$. I.e.,
$f\equiv\sum\limits_S\fhat(S)\chi_{_S}$. For a good introduction to
the subject see \cite{O'Donnell2008}, \cite{Wolf2008}.\\
In this proof we are using the following:
\begin{itemize}
\item $\chi_{_S}(xy)=\chi_{_S}(x)\chi_{_S}(y)$
\item $\Exp[\chi_{_S}(x)\chi_{_T}(x)]=\left\{\begin{array}{lr}1 & S=T\\0& otherwise\end{array}\right.$
\item $\Exp[f^2(x)]=\sum\limits_S\fhat^2(S)$
\item $\fhat(S)=1-2d(f,\chi_{_S})=2d(f,-\chi_{_S})-1$
\end{itemize}
} )
\begin{lemma}\label{mXOR_Main_Lemma}
Let
\mbox{$f^1,\ldots,f^m:\left\{-1,1\right\}^n\rightarrow\left\{-1,1\right\}$}
three voting functions and $\epsilon$ a constant such that
    \[\Pr\left[\prod\limits_{i=1}^{m-1}f^i(x^i)\neq f^m(\prod\limits_{i=1}^{m-1}x^i)\right]\leqslant\epsilon\]
Then,
\begin{itemize}
\item There exists a linear function
    \mbox{$\chi:\left\{-1,1\right\}^n\rightarrow\left\{-1,1\right\}$}
    defined as $\chi(x)=\prod\limits_{i\in S}x_i$ for some coalition
    $S$ and signs \mbox{$\left(a^i\right)_{i=1,\ldots,m}\in\{-1,1\}$} such that
    \[\begin{array}{r@{~\leqslant}l}
        \multicolumn{2}{c}{\prod\limits_{i=1}^{m}a^i=1}\\
        d(f^1,a^1\chi)&\epsilon\\
        \forall i~:~d(f^i,a^i\chi)&2\epsilon\\
    \end{array}\]
\item
If \mbox{$\epsilon<\frac{1}{6}$}, then there exists a linear
function
    \mbox{$\chi:\left\{-1,1\right\}^n\rightarrow\left\{-1,1\right\}$}
    defined as \mbox{$\chi(x)=\prod\limits_{i\in S}x_i$} for some
    coalition $S$ and signs \mbox{$\left(a^j\right)_{j=1,\ldots,m}\in\{-1,1\}$}
    such that \mbox{$\prod\limits_{j=1}^ma^j=1$} and
        \mbox{$d(f^j,a^j\chi)\leqslant \epsilon$} for all $j$
\end{itemize}
\end{lemma}
Noticing that $\left<\left(a^j\chi\right)\right>$ is a consistent
mechanism for any linear function $\chi$  and signs $a^j$ s.t.
$\prod\limits_{j=1}^ma^j=1$ gives us the requested result by
applying proposition \ref{Prop_FuncDistToMechDist}.
\begin{proof-of-lemma}{\ref{mXOR_Main_Lemma}}~\\
The main ingredient in the proof is the following lemma that
connects the inconsistency index with a simple expression over the
Fourier coefficients of $f^j$.
\begin{lemma}\label{Lemma_GeneralParseval}
Let \mbox{$f^1,\ldots,f^m:\left\{-1,1\right\}^n\rightarrow\Re$}.
Then:~~~ \mbox{$
    \Exp\left[\prod\limits_{j=1}^{m-1}f^j\left(x^j\right)f^m\left(\prod\limits_{j=1}^{m-1}x^j\right)\right]
        =\sum\limits_S\prod\limits_{j=1}^{m}\fhatj(S)$}.
\end{lemma}
\begin{corollary}\label{mXORcorollary}
For the aggregation mechanism $F=\left<f^1,\ldots,f^m\right>$:
\\\indent\mbox{$1-2IC(F)=\sum\limits_S\prod\limits_{j=1}^{m}\fhatj(S)$}.
\end{corollary}

\noindent Now let
\mbox{$F=\left<f^1,\ldots,f^m\right>$}
be an independent aggregation mechanism that satisfies
\mbox{$IC(F)\leqslant\epsilon$}.\\
First we claim there exists a coalition \mbox{$A$} and a sign
\mbox{$a^1\in\{-1,1\}$} s.t.
\mbox{$d\left(f^1,a\chi_{_{A}}\right)\leqslant\epsilon$}\\
$\begin{array}[t]{rll}%
1-2IC(F)
    &= \sum\limits_S\prod\limits_{j=1}^{m}\fhatj(S)
    &\leqslant \sum\limits_S\left|\widehat{f^1}(S)\right|\cdot\left|\prod\limits_{j=2}^{m}\fhatj(S)\right|\\
    &\leqslant \max\limits_S\left|\widehat{f^1}(S)\right|\sum\limits_S\prod\limits_{j=2}^{m}\left|\fhatj(S)\right|
    &\leqslant^{\mbox{Lemma \ref{Lemma_GeneralizedCauchySwartz}}}
        \max\limits_S\left|\widehat{f^1}\right|\prod\limits_{j=2}^{m}\sqrt{\sum\limits_S\left|\fhatj^{m-1}(S)\right|}\\
    &\leqslant \max\limits_S\left|\widehat{f^1}\right|\prod\limits_{j=2}^{m}\sqrt{\sum\limits_S\fhatj^2(S)}
        = \max\limits_S\left|\widehat{f^1}\right|
    &= 1-2\min\limits_{S,a\in\{-1,1\}}\left(d\left(f^1,a\chi_{_S}\right)\right)
\end{array}$\\

and hence there exists a coalition $A$ and a sign $a^1$ s.t.
\mbox{$\Pr[f^1(x)\neq a^1\chi_{_{A}}(x)]\leqslant IC(f,g,h)=\epsilon$}.

Based on proposition \ref{Prop_GenWisLifshitz},
    \mbox{$IC(a^1\chi_{_{A}}, f^2,\dots, f^m)\leqslant IC(F)+d(f,a^1\chi_{_{A}})\leqslant 2\epsilon$}.
On the other hand based on corollary
    \[IC(a^1\chi_{_{A}}, f^2,\dots, f^m)
        =\frac{1}{2}\left(1-a^1\sum\limits_S\widehat{\chi_{_{A}}}(S)\prod\limits_{j=2}^m\fhatj(S)\right)
        =\frac{1}{2}\left(1-a^1\prod\limits_{j=2}^m\fhatj(A)\right).\]
So we get that
\mbox{$a^1\prod\limits_{j=2}^m\fhatj(A)\geqslant1-4\epsilon$}

and hence there exist signs $(a^j)_{j=1}^m$ such that
\mbox{$\prod\limits_{j=2}^ma^j=1$} and
\mbox{$a^j\fhatj(A)\geqslant1-4\epsilon$} so
\mbox{$d(f^j,a^j\chi_{_{A}})\leqslant2\epsilon$}.

Due to symmetry there is also a coalition $B$ and a sign $b^2$
such that
    \mbox{$d(f^2,b^2\chi_{_B})\leqslant\epsilon$}
and hence
    \mbox{$d(b^2\chi_{_B},a^2\chi_{_A})\leqslant3\epsilon$}.
On the other hand
    \mbox{$d(b^2\chi_{_B},a^2\chi_{_A})=
        \left\{\begin{array}{ll}
            0 & a^2=b^2 \wedge A=B\\
            1 & a^2=b^2 \wedge A=B\\
            \half & A\neq B\\
        \end{array}\right.$}.\\
Hence, if \mbox{$\epsilon<\frac{1}{6}$}, we get that \mbox{$A=B$},
\mbox{$a^2=b^2$}.

Due to symmetry we can repeat this for all $f^j$.
\end{proof-of-lemma}
\end{proof}%
\subsection{Lemma \ref{Lemma_GeneralParseval}}%
\begin{lemmaUN}~\\
Let \mbox{$f^1,\ldots,f^m:\left\{-1,1\right\}^n\rightarrow\Re$}.
Then:~~~ \mbox{$
    \Exp\left[\prod\limits_{j=1}^{m-1}f^j\left(x^j\right)f^m\left(\prod\limits_{j=1}^{m-1}x^j\right)\right]
        =\sum\limits_S\prod\limits_{j=1}^{m}\fhatj(S)$}.
\end{lemmaUN}%
\begin{proof-of-lemma}{}~\\
$\begin{array}[t]{rlrr}%
\Exp_{x,y}[f(x)g(y)h(xy)]
    &= \Exp_{x,y}\left[\sum\limits_{S,T,R}\fhat(S)\chi_{_S}(x)\ghat(T)\chi_{_T}(y)\hhat(R)\chi_{_R}(xy)\right]\\
    &= \sum\limits_{S,T,R}\fhat(S)\ghat(T)\hhat(R)\Exp[\chi_{_S}(x)\chi_{_R}(x)]\Exp[\chi_{_T}(y)\chi_{_R}(y)]\\
    &= \sum\limits_S\fhat(S)\ghat(S)\hhat(S)&\qedhere
\end{array}$
\end{proof-of-lemma}%
\subsection{Lemma \ref{Lemma_GeneralizedCauchySwartz}}%
\begin{lemma}\label{Lemma_GeneralizedCauchySwartz}~\\
Let $k\geqslant2$ be an integer and $\{a_{i,j}\}_{i=1\ldots n,
j=1\ldots k}$ positive reals. Then,
\[\left(\sum\limits_{i=1}^n \prod_{j=1}^ka_{i,j}\right)^k\leqslant
    \prod\limits_{j=1}^k\left(\sum\limits_{i=1}^n\left(a_{i,j}\right)^k\right)\]
\end{lemma}%
\begin{proof-of-lemma}{}~\\
\begin{itemize}
\item Let $x_1,\ldots,x_k>0$ and $p_1,\ldots,p_k$ s.t. $\sum p_j=1$
then $\prod x_j^{p_j}\leqslant \sum p_jx_j$.\\
$-log(x)$ is convex and hence
    $-log(\sum p_jx_j)\leqslant -\sum p_j log(x_j) = -log(\prod
    x_j^{p_j})$.
$-log(x)$ is downward monotone and hence
    $\sum p_jx_j\geqslant\prod x_j^{p_j}$.

\item Let $y_1,\ldots,y_k>0$ and $q_1,\ldots,q_k$ s.t. $\sum \frac{1}{q_j}=1$
then $\prod y_j\leqslant \sum \frac{y_j^{q_j}}{q_j}$.\\
Assign $x_j\leftarrow y_j^{q_j}$ and $p_j\leftarrow \frac{1}{q_j}$
in the former.

\item Let $q_j,\ldots,q_k$ s.t. $\sum \frac{1}{q_j}=1$ and assume
    $\forall j~~:~~\sum\limits_{i=1}^n a_{i,j}^{q_j}=1$. Then
    $\sum\limits_{i=1}^n \prod_{j=1}^ka_{i,j}\leqslant
     \prod\limits_{j=1}^k\left(\sum\limits_{i=1}^n\left(a_{i,j}\right)^{q_j}\right)^{\frac{1}{q_j}}$\\
$\begin{array}{rl} \sum\limits_{i=1}^n \prod_{j=1}^ka_{i,j}
    &\leqslant \sum\limits_{i=1}^n \sum_{j=1}^k\frac{a_{i,j}^{q_j}}{q_j}\\
    &= \sum_{j=1}^k \frac{\sum\limits_{i=1}^n a_{i,j}^{q_j}}{q_j}\\
    &= \sum_{j=1}^k \frac{1}{q_j}\\
    &= 1\\
    &= \prod\limits_{j=1}^k\left(\sum\limits_{i=1}^n\left(a_{i,j}\right)^{q_j}\right)^{\frac{1}{q_j}}
\end{array}$

\item Let $q_1,\ldots,q_k$ s.t. $\sum \frac{1}{q_i}=1$. Then
    $\sum\limits_{i=1}^n \prod_{j=1}^ka_{i,j}\leqslant
     \prod\limits_{j=1}^k\left(\sum\limits_{i=1}^n\left(a_{i,j}\right)^{q_j}\right)^{\frac{1}{q_j}}$\\
Normalize each vector. (The case of zero is trivial)

\item $\left(\sum\limits_{i=1}^n \prod_{j=1}^ka_{i,j}\right)^k\leqslant
    \prod\limits_{j=1}^k\left(\sum\limits_{i=1}^n\left(a_{i,j}\right)^k\right)$\\
Take $q_i=k$
\end{itemize}
\end{proof-of-lemma}%

\subsection{Affine Agenda - Lemma \ref{Lemma_AffineIsXOR}}
\begin{lemma}\label{Lemma_AffineIsXOR}
Let \FeasableAssignments be an affine subspace of $\{0,1\}^m$ of degree $k$.\\
Then \Agenda can be represented as a truth-functional agenda using xor conclusions only.
\end{lemma}

\begin{proof-of-lemma}{}~\\
\FeasableAssignments is an affine space and therefore can be
represented as a linear subspace shifted by a constant vector.
Shifting is merely renaming of the opinions so with no loss of
generality, assume that \FeasableAssignments is a linear subspace
defined by a matrix $A_{k\times m}$ of rank $k$ in the following way
$\FeasableAssignments=\{x\in\{0,1\}^m~|~Ax=0\}$. There exists an
invertible matrix (representing the Gaussian elimination process)
$P$ s.t.
\begin{itemize}
\item $\{x\in\{0,1\}^m~|~Ax=0\}= \{x\in\{0,1\}^m~|~PAx=0\}$
\item $PA$ is in canonical form. I.e. for any row $t\in[k]$ there is a unique index $a_t\in[m]$ s.t. $(PA)_{t,j}=1$ iff $j=a_t$.
\end{itemize}
Hence \FeasableAssignments is a truth-functional agenda for the
premises $[m]\setminus\{a_t\}_{t\in[k]}$ and conclusions based on
the row of $PA$.
\end{proof-of-lemma}
\end{document}